\documentclass[a4paper,UKenglish,cleveref, autoref, thm-restate]{lipics-v2021}

\hideLIPIcs

\usepackage[dvipsnames]{xcolor}
\usepackage{tikz}
\usetikzlibrary{calc,shapes}
\usepackage[utf8]{inputenc}
\usepackage{amsmath,amsthm,amssymb}
\usepackage{mathtools}
\usepackage{diagbox}
\usepackage{todonotes}
\usepackage[unicode]{hyperref}
\usepackage{xspace}
\usepackage{thm-restate}
\usepackage[noend]{algpseudocode}
\usepackage{caption}
\usepackage{subcaption}
\usepackage{algorithm}
\usepackage{graphicx}

\usepackage{thm-restate}

\newcommand{\RR}{\mathbb{R}}
\newcommand{\NN}{\mathbb{N}}
\newcommand{\PP}{\mathcal{P}}
\newcommand{\OO}{\mathcal{O}}

\DeclareMathOperator*{\argmax}{arg\,max}

\newcommand{\kk}{\texorpdfstring{$k$}{k}}
\newcommand{\NP}{\texorpdfstring{$\mathcal{NP}$}{NP}}

\newcommand{\Ax}{\ensuremath{A_X}\xspace}
\newcommand{\Px}{\ensuremath{P_X}\xspace}
\newcommand{\Am}{A_\mathrm{total}}
\newcommand{\Pm}{p_{\mathrm{max}}}
\newcommand{\Pt}{P_{\mathrm{total}}}

\newcommand{\poly}{\mathcal{T}}

\DeclarePairedDelimiter\floor{\lfloor}{\rfloor}
\newcommand{\round}[1]{\ensuremath{\lfloor#1\rceil}}

\preto\align{\setcounter{equation}{0}%
             }

\nolinenumbers

\title{Finding Regions of Maximum Circularity in Plane Geometric Graphs}

\author{Jan-Henrik Haunert}
{Institute of Geodesy and Geoinformation, University of Bonn, Germany \and \url{https://www.uni-bonn.de/de/forschung-lehre/forschungsprofil/transdisziplinaere-forschungsbereiche/modelling/mitgliederverzeichnis/jan-henrik-haunert} }
{haunert@igg.uni-bonn.de}
{https://orcid.org/0000-0001-8005-943X}
{}
\author{Joshua Marc K\"onen}{Institute of Computer Science, University of Bonn, Germany \and \url{https://nerva.cs.uni-bonn.de/doku.php/staff/joshuakoenen}}{koenen@cs.uni-bonn.de}{https://orcid.org/0000-0003-4245-4812}{}
\author{Heiko R\"oglin}{Institute of Computer Science, University of Bonn, Germany \and \url{https://www.roeglin.org/}}{roeglin@cs.uni-bonn.de}{https://orcid.org/0009-0006-8438-3986}{}
\author{Tarek Stuck}{Institute of Computer Science, University of Bonn, Germany}{s6tastuc@uni-bonn.de}{}{}

\funding{This work has been funded by the Deutsche
Forschungsgemeinschaft (DFG, German Research Foundation) – 390685813;
459420781 and by the Lamarr Institute for Machine Learning and
Artificial Intelligence lamarr-institute.org.}

\authorrunning{J. Haunert, J.\,M. Könen, H. Röglin and T. Stuck}

\Copyright{Jane Open Access and Joan R. Public}



\category{}

\relatedversion{} 

\EventEditors{John Q. Open and Joan R. Access}
\EventNoEds{2}
\EventLongTitle{42nd Conference on Very Important Topics (CVIT 2016)}
\EventShortTitle{CVIT 2016}
\EventAcronym{CVIT}
\EventYear{2016}
\EventDate{December 24--27, 2016}
\EventLocation{Little Whinging, United Kingdom}
\EventLogo{}
\SeriesVolume{42}
\ArticleNo{23}

\begin{document}

\maketitle

\begin{abstract}
A problem that occurs in different applications in geographical information science is to generate compact regions from areas on a map. 
This is important, e.g., in the context of electoral districting to avoid gerrymandering. 
A common measure for the compactness of a region is the Polsby-Popper score, which measures how close a given region is to a circle based on its area and perimeter. 
We assume that a polygonal subdivision of the plane is given and study the problem of selecting a subset of the polygonal faces that maximizes the Polsby-Popper score, given by $\frac{4\pi A}{P^2}$, where $A$ is the area of the selected shape and $P$ is its perimeter.

We consider the more general task of maximizing $\frac{A}{P^\alpha}$ for an arbitrary $\alpha>1$, which we call the $\alpha$-circularity problem.
We perform the first rigorous study of its complexity and show that it is weakly NP-hard if $\alpha \in (1,2]$. Furthermore, for $\alpha>1$ we present a pseudopolynomial time algorithm for this problem.

\end{abstract}

\section{Introduction}

In geographical information science and spatial planning, \emph{spatial allocation} refers to the determination of regions by selecting sets of faces from a given subdivision of the plane~\cite{Shirabe2005}. 
Often the faces correspond to administrative areas or cells of a regular grid. 
The version of the problem where the set of faces needs to be partitioned into multiple regions is commonly referred to as \emph{districting}, which is particularly relevant for defining electoral districts~\cite{ValidiBL2022}. 
Another important case of spatial allocation deals with determining a single region, e.g., to locate a nature reserve or a hazardous waste site \cite{CovaC2000}.
While spatial allocation often follows application-specific criteria, some criteria are of general relevance. These include criteria related to size, compactness, and contiguity of the output regions. 
While contiguity is often imposed as a hard constraint, compactness is usually modeled as an optimization objective.
In graph-based redistricting models, a common compactness objective is to minimize the number of cut edges, i.e., adjacency edges whose endpoints are assigned to different districts~\cite{DDS19}.
Here, we study the special case of spatial allocation that asks for selecting a single set of faces maximizing the compactness of the region resulting from their union. 
To measure compactness, we use the well-established Polsby-Popper score, which is frequently applied in electoral districting to avoid gerrymandering, i.e., the manipulation of electoral districts to advantage a particular party \cite{PolsbyP1991}. Although this score cannot fully prevent gerrymandering, particularly in cases of \textit{stealth gerrymanders}, where geometric regularity is preserved while electoral outcomes remain skewed (see, e.g., \cite{BLSS22,CG20}), it remains a commonly used tool for detecting political manipulation.
The problem that asks for a single region of maximum compactness can also be understood as a pattern recognition task. 
It is of relevance, e.g., for detecting ring roads in road networks, which has been studied in the context of cartographic generalization~\cite{HeinzleAS2006}
and a modified version has been used in medical image analysis as a numerical parameter to predict cervical lymph node metastases, where non-compact regions are associated with more aggressive tumor growth and infiltration~\cite{CCBKRMSIEG24}.

The Polsby-Popper score of an arbitrary shape in the plane is defined as $\frac{4\pi A}{P^2}$, where $A$ is the area of the shape and $P$ the length of its perimeter. 
From the isoperimetric inequality it is known that for a fixed length $P$ of a closed curve we have the inequality $4\pi A \leq P^2$, where $A$ is the enclosed area. Additionally, equality holds if and only if the closed curve corresponds to a circle. Therefore the Polsby-Popper score always falls within the range $[0,1]$.
By definition, it is also scale-invariant, because enlarging or shrinking all polygons by the same ratio does not affect its value.
Although the Polsby-Popper score is among the most popular compactness measures for shape analysis and spatial allocation, it has been criticized for being sensitive to the sinuosity of the region boundary and the level of detail at which it is represented
\cite{BarnesS2021}. 
Other compactness measures based on moments of inertia or comparisons of a region with a reference shape (e.g., the region’s convex hull) are therefore often preferred
\cite{Maceachren1985}. 
However, we deem the Polsby-Popper score a reasonable measure for detecting regions whose boundary is short by design, e.g., regions enclosed by circular ring roads.

The Polsby-Popper score has been analyzed for a long time with results even coming up in recent years. 
Recently, it was analyzed by Belotti et al.~\cite{BBE24}, who introduced mixed-integer second-order cone programs (MISOCP) for finding either a single district with optimal Polsby-Popper score or multiple districts such that the average score of the districts is maximized.
In the following we will refer to the first problem as the \textit{circularity} problem.
Using a MISOCP formulation also allows the inclusion of additional constraints for either problem, such as population balance or lower and upper bounds on the population of each district.

We will consider a more general version of the circularity problem, where instead of finding a set of polygons maximizing $\frac{4\pi A}{P^2}$, we seek to maximize $\frac{A}{P^\alpha}$ for an arbitrary $\alpha > 1$. We refer to this problem as the $\alpha$-\textit{circularity problem}. 
Note that the larger the parameter $\alpha$ is chosen, the more focus is put on minimizing the perimeter. 
This means that for large values of $\alpha$, smaller regions are favored, while for $\alpha$ close to $1$, we focus more on the enclosed area. \cref{fig:example_instance_alpha} illustrates the influence of $\alpha$ on the optimal solution for a given instance.
Additionally, we will consider the problem of finding exactly $k$ districts in a map such that the minimum $\alpha$-circularity among all districts is maximized.
We refer to this problem as the \emph{$k$-districting} problem.
To the best of our knowledge, we are the first to analyze the theoretical complexity of the circularity problem and $k$-districting problem.

\subsection{Related Work}
An optimization problem similar to the circularity problem is the \textit{Minimum Quotient Cut} problem. 
In this problem we are given a graph $G=(V,E)$ as well as edge costs $c:E \rightarrow \mathbb{R}_+$ and vertex weights $w:V \rightarrow \mathbb{R}_+$. The goal is to find a cut $S \subset V$ minimizing $\frac{\text{cost}(S)}{\min\{w(S), w(V\setminus S)\}}$, where $\text{cost}(S)$ corresponds to the cost of all edges with one endpoint in $S$ and one endpoint in $V\setminus S$, i.e., $\text{cost}(S)=\sum_{e=\{u,v\}:u\in S, v\in V\setminus S}c(e)$. 

A related problem is the \emph{Sparsest Cut} problem
where we have a supply graph and a demand graph and the cost of a cut is measured by the total cost of the edges that are being cut in the supply graph divided by the total cost of the edges that are being cut in the demand graph.
For the case of uniform demands, where every pair of vertices has unit demand, Arora, Rao and Vazirani gave an $\mathcal{O}(\sqrt{\log n})$-approximation algorithm~\cite{ARV09}. 
For general demands, an algorithm by Arora, Lee and Naor achieves an approximation factor of $\mathcal{O}(\sqrt{\log n } \log \log n)$~\cite{ALN05}.
Regarding hardness, for the non-uniform sparsest cut there is an approximation hardness of $17/16 - \varepsilon$, even when the supply graph has treewidth at most 2~\cite{GTW13}.

For planar graphs, the minimum quotient cut problem can be solved in pseudopolynomial time because minimal cuts in the primal graph correspond to simple cycles in the respective dual graph.
In this setting, the first exact algorithm was presented by Park and Phillips, who solved the planar minimum quotient cut problem in pseudopolynomial time $\tilde{\mathcal{O}}(n^2W)$, where $W$ is the sum of the vertex weights~\cite{PP93}. Additionally, they showed that the problem is weakly \NP{}-hard for 2-outerplanar and series-parallel graphs.

There is a natural representation of the $\alpha$-circularity problem that almost corresponds to the setting described above:
We construct the adjacency graph of our polygon instance where every vertex $v_p$ represents a polygon $p$, there is one vertex $z$ for the outer face, and every vertex $v_p$ has a weight equal to the area of $p$.
Two vertices are connected by an edge if the corresponding polygons have a common boundary. The cost of the edge is set to the length of this boundary, i.e., the total length of all polygon edges between the two polygons.

In this graph, we are interested in finding a set of vertices $S\subseteq V\setminus \{z\}$ that maximizes the ratio $\frac{w(S)}{\text{cost}(S)^\alpha}$. 
Although this formulation is not identical to the minimum quotient cut problem, it enables us to use algorithms developed for minimum quotient cut in our setting.
Specifically, we utilize the search graph, as defined in~\cite{PP93}, to compute the solution that maximizes the $\alpha$-circularity for any given $\alpha>1$. 

\subsection{Our Results}
In \cref{sec:Alg} we show that for any $\alpha>1$, an optimal subset of polygons with area $A$ and perimeter $P$ that maximizes $A/P^\alpha$ can be computed in pseudopolynomial time $\OO(n^3 \cdot m \cdot \Am)$, where $\Am$ is the sum of areas of all polygons in the instance, assuming that all polygons have integer area values, and $n,m$ are the total number of vertices and edges of the polygons, respectively. 
Alternatively, if all polygon edge lengths are integers, an optimal solution can be computed in pseudopolynomial time $\OO(n^2\cdot m\cdot \Pt)$, where $\Pt$ is the total sum of the edge lengths of all polygons.

The circularity problem can also be viewed as a bicriteria optimization problem, where the goal is to minimize the perimeter and to maximize the area of a solution.
In such bicriteria optimization settings, one is often interested in computing the set of Pareto-optimal solutions, where a solution is Pareto-optimal if no other solution is at least as good in both criteria and strictly better in at least one. An important observation is that any solution with maximum circularity must be Pareto-optimal because otherwise there would exist a solution with larger circularity. Our pseudopolynomial time algorithm computes the set of Pareto-optimal solutions by dynamic programming and then picks the solution with largest circularity among them. The running time of this approach can be bounded by $\OO(n^2\cdot m\cdot p_{\max})$, where $p_{\max}$ corresponds to the largest Pareto set arising in any subproblem solved by the dynamic program. If the areas or edge lengths are integers, $p_{\max}$ can be bounded from above by $n\cdot\Am$ and $\Pt$, respectively, resulting in the running times stated above.

Let us point out that our algorithm also works if the areas and edge lengths are rational or real-valued, and also in this case it achieves the claimed runtime of $\OO(n^2\cdot m\cdot p_{\max})$. We used integer values only to bound the size $p_{\max}$ of the largest Pareto set in the worst case. In general, this is a rather pessimistic estimate. While, in the worst case, Pareto sets can be of exponential (or pseudopolynomial) size for almost all bicriteria optimization problems, there are usually much fewer Pareto-optimal solutions on real-world instances (see, e.g.,~\cite{Ehrgott2005,Muller-HannemannW06}). This is not only observed in experiments but it has also been shown in theory that the expected number of Pareto-optimal solutions is polynomial in the probabilistic input model of smoothed analysis (see, e.g.,~\cite{BrunschR15,BeierRRV2022}). We demonstrate experimentally on real-world inputs that this behavior is also observed for the circularity problem, and that the actual value of $p_{\max}$ is typically far smaller than its worst-case bound.

For instances with rational area values, we show in \cref{sec:FPTAS} that rounding to appropriate integer values yields an FPTAS with runtime $\OO(\frac{n^3\cdot m^2\cdot f(m)}{\varepsilon})$, assuming 
$A_{\text{total}}/A_{\text{min}}$
is bounded by some polynomial function $f(m)$, where $A_{\text{min}}$ corresponds to the smallest area in the instance.  

In \cref{sec:NP-hardness} we show that solving the $\alpha$-circularity problem is weakly \NP{}-hard for any $\alpha \in (1,2]$ by a reduction from the weakly \NP{}-complete \textit{partition} problem.
The reduction can also be adapted to show \NP{}-hardness of a similar problem:
In \cref{sec:NP-hardness-districting} we show that assigning each polygon to one of $k$ districts such that each district contains at least one polygon and the smallest $\alpha$-circularity score is maximized is also weakly \NP{}-hard for any $\alpha \in (1,2]$. 

\begin{figure}%
    \centering
     \begin{subfigure}[t]{0.32\textwidth} 
         \centering
         \includegraphics[width=\textwidth, page=11]{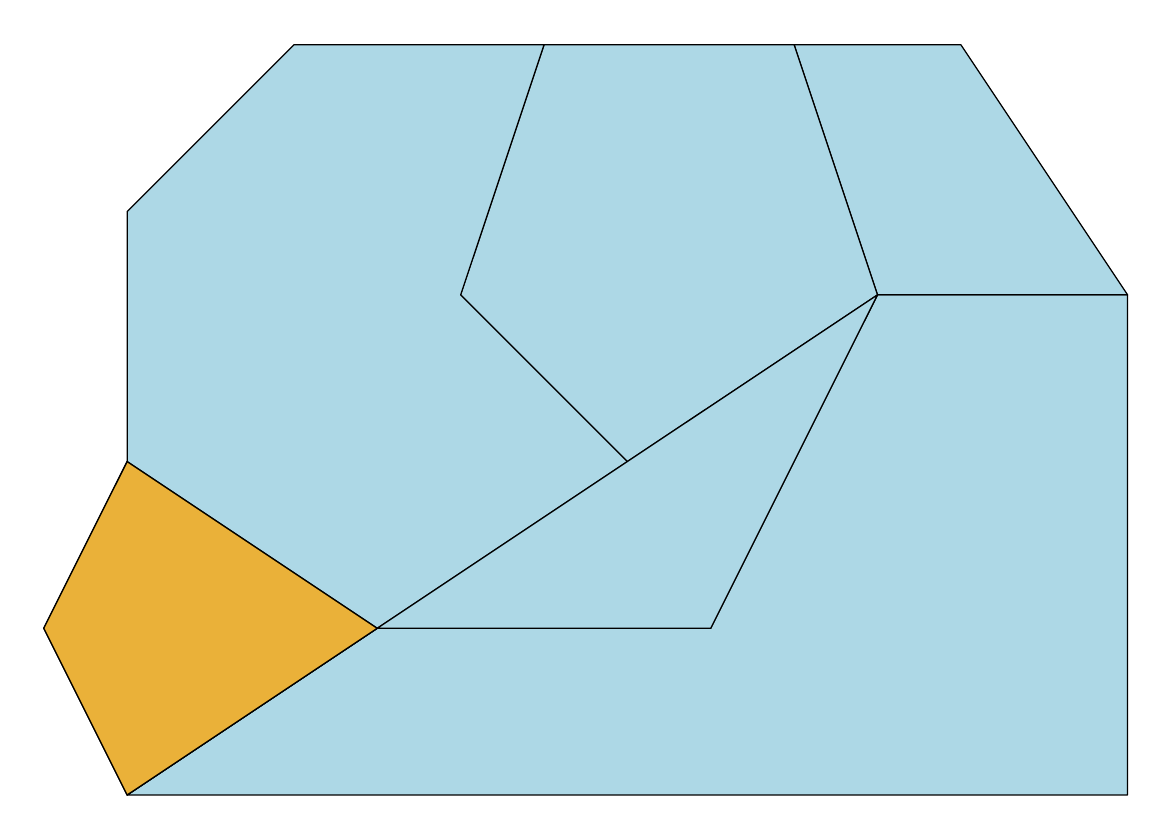} 
         \subcaption{Optimal region for $\alpha=3$.}
     \label{fig:example_left}
     \end{subfigure}
     \hfill
     \begin{subfigure}[t]{0.32\textwidth}
         \centering        
\includegraphics[width=\textwidth, page=12]{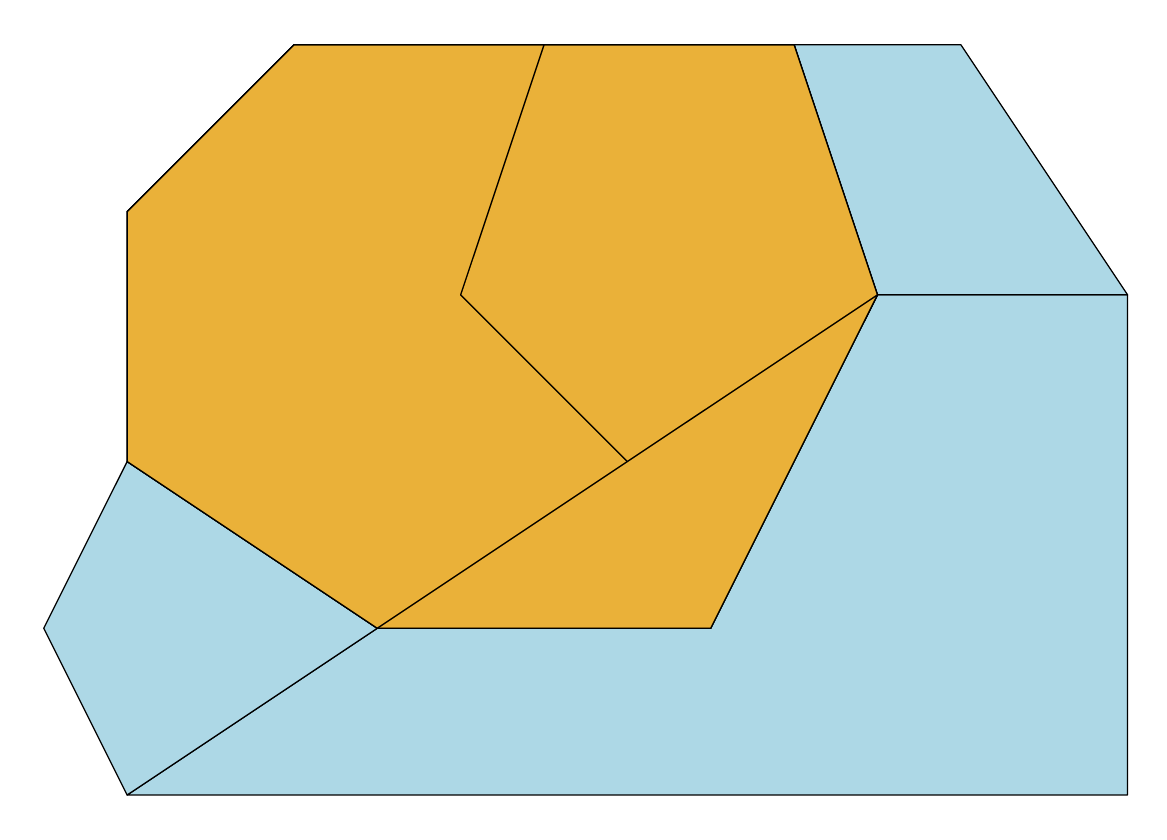} 
         \subcaption{Optimal region for $\alpha = 2$.}
         \label{fig:example_middle}
     \end{subfigure}
    \hfill
    \begin{subfigure}[t]{0.32\textwidth}
        \centering
        \includegraphics[width=\textwidth, page=13]{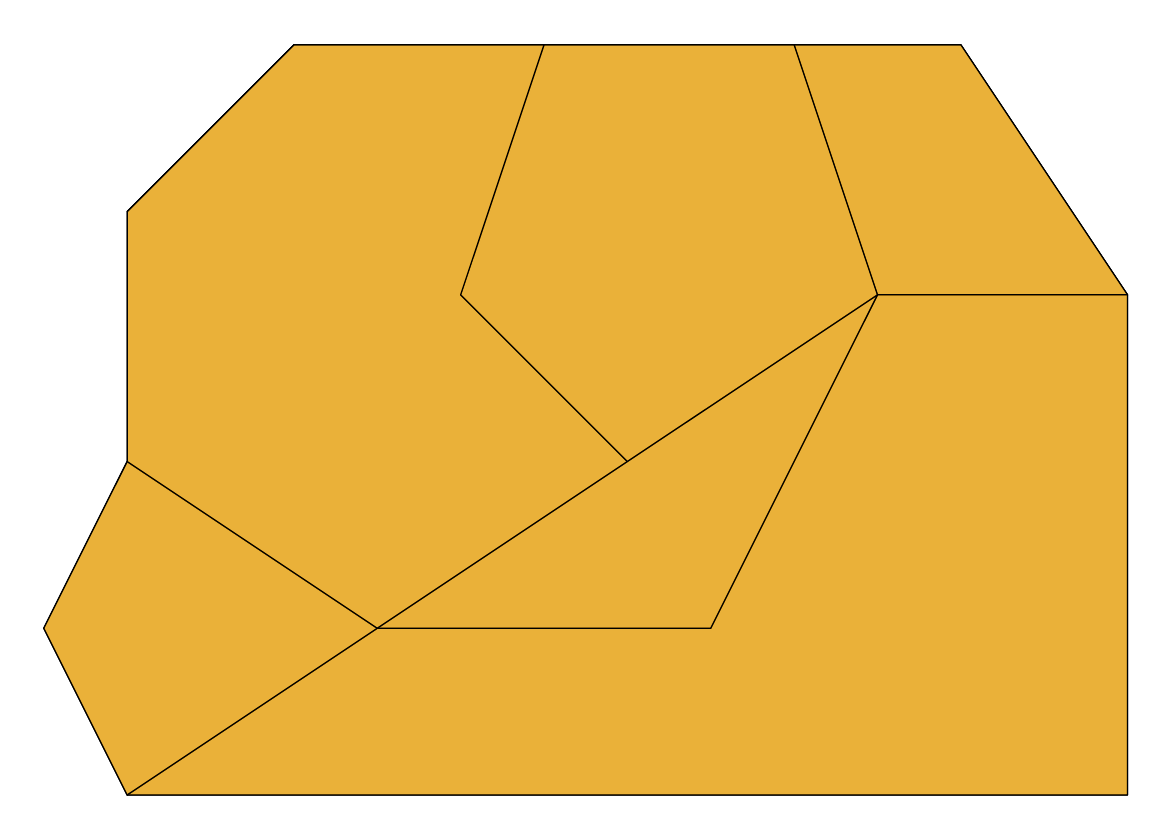} 
        \subcaption{Optimal region for $\alpha=1.1$.}
        \label{fig:example_right}
    \end{subfigure}
    \caption{Example instance with edge lengths and areas induced by a Euclidean embedding. Optimal solutions maximizing $\alpha$-circularity for different values of $\alpha$ are highlighted in orange.
    }
    \label{fig:example_instance_alpha}
\end{figure}

\section{Preliminaries} \label{sec:Preliminaries}

Let $[i]:= \{1,\ldots,i\}$.
We aim to find a subset of polygons whose union forms the overall most compact shape.
By the isoperimetric inequality, the most compact shape in two dimensions is a circle. 
Therefore, given a set of polygons $\poly$, we aim to find a subset $S\subseteq \poly$ whose union maximizes $\frac{4\pi A(S)}{P(S)^2}$, where $A(S)$ and $P(S)$ denote the area and perimeter of the union of all polygons in $S$, respectively.
We assume that the set of polygons $\poly$ is connected and does not contain any holes.
Since constant factors do not affect the relative ordering of solutions in terms of quality, we can simplify this expression to $\frac{A(S)}{P(S)^2}$.
We define $c_\alpha(S):=\frac{A(S)}{P(S)^\alpha}$ as our general $\alpha$-circularity objective function, where $\alpha = 2$ corresponds to the original Polsby-Popper score.
We assume that all polygons have positive area and positive edge lengths.

A polygonal instance is given by $G_\poly=(G=(V,E),\ell,A)$, where $G$ is a Euclidean plane graph with $n=\vert V \vert$ vertices and $m = \vert E \vert$ edges representing the vertices and edges of the polygons $\poly$.
For any subset of the polygons $S\subseteq \poly$ there exists a unique edge set $E_S\subseteq E$ that contains all boundary edges of the union of the polygons in set $S$.
The function $\ell:E\rightarrow \mathbb{R}_{>0}$ assigns a length to each edge. Moreover, the function $A: 2^E\rightarrow \mathbb{N}_{0}$ assigns to each edge set that corresponds to the boundary of the union $S$ of some polygons its enclosed area and 0 otherwise.
For $E_S\subseteq E$, we use the shorthand $\ell(E_S)$ to denote $\sum_{e\in E_S}\ell(e)$.

The $\alpha$-circularity score of $S$ is given by $\frac{A(E_S)}{\ell(E_S)^\alpha}$.
For simplicity, we will use the function $A$ 
interchangeably in the context of a polygon instance $\poly$ or its corresponding graph representation $G_\poly$. That is, we  write $A(S)$ if $S \subseteq \poly$ is a subset of the polygons and $A(E_S)$ when considering the edge set in the graph $G_\poly$. 

First we observe that, for any $\alpha \geq 0$, an optimal solution cannot contain holes because filling these holes strictly increases the area and strictly decreases the perimeter.

\begin{observation}\label{obs:hole}
For any $\alpha \geq 0$, an optimal set of polygons $S$ maximizing $c_\alpha(S)$ does not contain holes.
\end{observation}
Next, we show that for any $\alpha > 1$ the polygons in an optimal circularity solution must be connected.
This result is similar to Theorem~2.2 in~\cite{PP93}, showing that an optimal solution for the minimum quotient cut is always achieved by a simple cycle that traverses the boundary of some region in the corresponding dual graph.

First, we will consider the case where the solution consists of two disconnected polygons.
The following lemma follows from elementary calculations, its proof is deferred to \cref{apx:singlecycleproof}.

\begin{restatable}{lem}{SingleCyclex}\label{lemma:SingleCycle_x}
If $S=S_1\cup S_2$ for two edge-disjoint and non-intersecting sets of polygons $S_1,S_2\subseteq \poly$ with $A(S_1),A(S_2)>0$, 
we have $c_\alpha(S)<\max\{c_\alpha(S_1),c_\alpha(S_2)\}$ for every $\alpha>1$.
\end{restatable}

We can extend this result to any collection of disjoint regions. 
Let $S=\bigcup_{i\in [k]}S_i$, where all $S_i$ are pairwise disjoint and satisfy $A(S_i)>0$.
Partitioning $S$ into $S_k$ and $\bigcup_{i\in [k-1]}S_i$, together with \cref{lemma:SingleCycle_x} implies $c_\alpha(S)<\max\{c_\alpha(S_k),c_\alpha(\bigcup_{i\in [k-1]}S_i)\}$.
Therefore, at least one of these two sets strictly improves the objective value while containing fewer components. Repeating this process yields a single region $\hat{S}=S_j$ for some $j\in [k]$.
As a consequence, we obtain the following observation.
\begin{observation}\label{obs:MultipleCycles}
An optimal set of polygons $S$ that maximizes $c_\alpha(S)$ for $\alpha > 1$ can be represented by a single simple cycle that traverses along the boundary of the selected polygons.
\end{observation}

In general, the number of vertices $n$ and edges $m$ in $G_\poly$ might not be bounded in terms of the number of polygons $\vert \poly\vert$, as the boundary of a polygon might contain arbitrarily many vertices with degree 2. 
To ensure a linear bound in $\vert \poly \vert$, we can contract vertices with degree 2 by merging their incident edges as long as this does not result in a multigraph.
The following lemma is deferred to \cref{apx:PolyBound}.

\begin{restatable}{lem}{PolyBound}\label{lemma:PolyBound}
Contracting as many vertices with degree 2 as possible yields an instance $G_\poly$ with $n=\OO(\vert \poly \vert)$ and $m=\OO(\vert \poly \vert)$.
\end{restatable}

\section{Algorithm for the \texorpdfstring{$\alpha$}{alpha}-circularity Problem}\label{sec:Alg}
From \cref{obs:MultipleCycles} we know that for any $\alpha > 1$ an optimal solution must be a simple cycle.
Hence, we design an algorithm that 
computes the set of Pareto-optimal cycles in the graph $G_{\poly}$
with respect to maximizing the enclosed area and minimizing the perimeter.

Our algorithm builds on the general approach of Park and Phillips~\cite{PP93} for solving the planar minimum quotient cut problem in pseudopolynomial time. 
Their method can be adapted with only minor modifications to solve the $\alpha$-circularity problem for any $\alpha > 1$. However, this direct adaptation has two significant drawbacks: it applies only to integral areas, and its running time is quasilinear in the total area $\Am$. Moreover, because of the algorithm’s structure, this dependence on $\Am$ occurs for every instance and is therefore not merely a worst-case bound. We propose a modified algorithm that also supports real-valued areas and edge lengths. Although its worst-case running time is likewise pseudopolynomial on instances with integral areas or edge lengths, its performance on realistic instances can be considerably better. We discuss this in more detail after presenting the algorithm and report experimental results confirming this behavior.

Before explaining our algorithm, we explain the approach by Park and Phillips~\cite{PP93} for solving the planar minimum quotient cut problem in pseudopolynomial time and how we modify this approach to solve the $\alpha$-circularity problem.
Their algorithm first constructs for the given planar embedded graph $G$ with edge costs and vertex weights its dual graph $G_D=(V_D,E_D)$
and then creates a directed search graph $G_s=(V_s,E_s)$ from $G_D$, where edges have lengths and weights. 
The search graph $G_s$ is constructed by fixing an arbitrary spanning tree $T$ of $G_D$ and replacing every undirected edge $\{x,y\}\in E_D$ by two directed edges $(x,y),(y,x)$. 
Each directed edge $e$ in $E_s$ is associated with a length $\ell'(e)$ and a weight $w(e)$. 
For edges $e=(x,y)$ with $\{x,y\}\in E(T)$, the weight is defined as $w(e)=0$.
For all other edges $e$, the weight $|w(e)|$ corresponds to the weight of the region enclosed by $e$ and the unique $y$-$x$ path in $T$. 
Depending on the orientation of $e$, we associate $w(e)$ with either the positive or negative value of the weight enclosed together with the tree $T$. 
We assume that $w(e)$ is positive if the enclosed weight lies to the left of $e$ and negative otherwise.
\cref{fig:search_graph_weights} shows an example of the weights in the search graph.
The construction can also be interpreted as choosing a fundamental cycle basis of $G_D$ induced by the spanning tree $T$ and then defining the weight of the directed non-tree edges based on the orientation and enclosed weight of the respective fundamental cycle.

In our application, a polygon instance $G_\poly$ corresponds to the dual graph $G_D$ of a corresponding planar minimum quotient cut instance.
The faces of $G_\poly$ correspond to the vertices of the original minimum quotient cut instance, and their areas correspond to the respective vertex weights.
Consequently, the search graph $G_s$ and the spanning tree $T$ can be obtained directly from $G_\poly$.

The following lemma, used in \cite{PP93}, shows that the weight of any simple cycle corresponds to its enclosed weight. 
Since its proof is omitted in~\cite{PP93}, we provide a proof in \cref{apx:weight_enclosed} for completeness.
\begin{restatable}[\cite{PP93}]{lemma}{CycleEnclosingLemma}\label{lem:weight_enclosed}
For every simple cycle $C$ in a polygon instance $G_\poly$ with enclosed area $A(C)$ and length $\ell(C)$, for the search graph $(G_s,\ell',w)$ we have $\vert w(C)\vert = A(C)$ and $\ell'(C)=\ell(C)$.
\end{restatable}

Park and Phillips use the search graph to find an optimal minimum quotient cut in pseudopolynomial time with respect to the sum of the vertex weights as follows: 
From the search graph an expanded graph is constructed, where every vertex additionally gets associated with a weight, ranging from $-W$ to $W$ with $W$ corresponding to the sum of all weights, i.e., a vertex $x\in V_s$ is replaced by vertices $x_c=(x,c)$ for $c\in \{-W,\dots,W\}$.
Similarly, for any edge $(x,y)\in E_s$ with weight $h$, the edge $((x,c),(y,c+h))$ is introduced for any $c\in \{-W,\ldots, W\}$ (if $\vert c+h\vert>W$, these edges are not required).
The length of each edge is set to the original length as in the search graph.
Now one can show that finding the optimal minimum quotient cut in the original graph is equivalent to finding a shortest path from some starting vertex $(x,0)$ to some $(x,b)$ for $b\in \{-W,\dots,W\}$ that optimizes their so-called pseudoquotient, defined as $\ell(C)/\min\{\vert w(C)\vert,W-\vert w(C)\vert\}$.

Next we discuss the modifications to the algorithm from Park and Phillips that are used for our algorithm. 
Let $\alpha > 1$ be fixed and $C_\alpha$ be a directed simple cycle that traverses the boundary of an optimal subset of polygons with respect to $c_\alpha$ in counterclockwise order.
First, we guess a vertex $u$ on this cycle.
Our algorithm does not construct the expanded graph but directly works on the search graph instead. 
In the search graph we compare paths from $u$ to an arbitrary vertex $v$ based on their respective length and weight.
For a given path $s$ in the search graph, its length and area correspond to the sum of the respective length and weight values of the edges contained in $s$.
Thus, a feasible solution (whether an intermediate path or the final cycle) is represented as $s\in \mathbb{R}^2$ with bicriteria cost $s=(\ell_s,w_s)$, where $\ell_s$ represents the path length and $w_s$ encodes the corresponding area.
Our goal is to minimize the first criterion while maximizing the second.

We say that a path $s_1$ \textit{dominates} another path $s_2$ if their endpoints are identical and $\ell_{s_1} \leq \ell_{s_2}$ and $w_{s_1} \geq w_{s_2}$, with at least one inequality being strict. 
For brevity, we write $s_1 \prec s_2$ in this case. 
A path $s$ is \emph{Pareto-optimal} if there does not exist another path $s'$ such that $s'\prec s$. 
The set of all paths that are not dominated by any other path is called the \emph{Pareto set}. 
Let $\PP_1, \PP_2$ be two sets of Pareto-optimal paths.
We define their \emph{combined} Pareto set as $\PP_1 \oplus \PP_2 := \{p \in \PP_1 \cup \PP_2 \mid p \text{ is Pareto-optimal in } \PP_1 \cup \PP_2\}$, which corresponds to all Pareto-optimal paths that are not dominated by any path $s \in \PP_1\cup \PP_2$.

We use the modified bicriteria graph $G_s$ to compute all Pareto-optimal cycles with at most $n$ edges. Since we cannot restrict the algorithm to only compute simple cycles (which follows from the NP-hardness of the longest path problem), the computed Pareto set will also contain non-simple cycles. 
Consequently, solutions might correspond to walks that traverse edges multiple times.
However, only simple cycles correspond to valid polygon selections in the original instance. Our algorithm selects among the computed Pareto set the solution that maximizes $c_\alpha$. We will show that this will always be a simple solution that also globally maximizes $c_\alpha$. 
\begin{figure}%
    \centering
     \begin{subfigure}[t]{0.3\textwidth} 
         \centering
         \includegraphics[width=\textwidth, page=6]{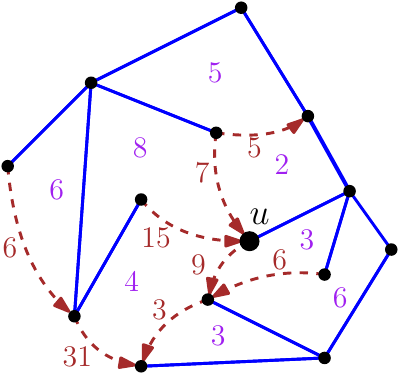} 
         \subcaption{Example polygonal instance.}
         \label{fig:search_graph_weights}
     \end{subfigure}
     \hfill
     \begin{subfigure}[t]{0.3\textwidth}
         \centering        
\includegraphics[width=\textwidth, page=7]{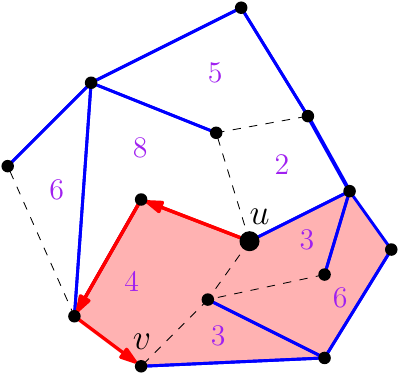} 
         \subcaption{Example path cost.}
         \label{fig:pareto_path}
     \end{subfigure}
    \hfill
    \begin{subfigure}[t]{0.3\textwidth}
        \centering
        \includegraphics[width=\textwidth, page=8]{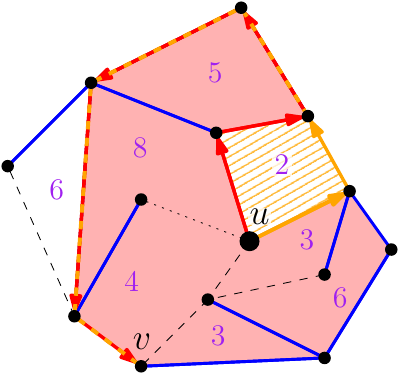} 
        \subcaption{Dominance behavior.}
        \label{fig:dom_pareto_path}
    \end{subfigure}
    \caption{%
    An example polygon instance where area values are shown in violet. \\
    \cref{fig:search_graph_weights}: The spanning tree is highlighted in blue and all positive directed edges in dashed brown. \\    
    \cref{fig:pareto_path}: 
    Exemplary Pareto-optimal path, highlighted in red, with associated weight $16$.\\
    \cref{fig:dom_pareto_path}: 
    The red and orange paths have same length but the orange path has larger associated weight and therefore dominates the red path.
    }
    \label{fig:example_pareto_paths}
\end{figure}

We first construct $G_s$ from $G_\poly$. For each node $x\in V(G_s)$ we then compute the set of Pareto-optimal cycles that contain $x$ using the bicriteria Bellman-Ford algorithm~\cite{CM85}. For every node $y\in V(G_s)$ and $i\in \{0,\ldots,n\}$, this algorithm maintains a set $L_y^i$ of Pareto-optimal paths from $x$ to $y$ that consists of at most $i$ edges. 
In particular, $L_x^i$ contains all Pareto-optimal paths from $x$ to itself with at most $i$ edges.
If the cycle that maximizes $c_\alpha$ contains $x$, the corresponding solution is guaranteed to be contained in $L_x^n$. 
See \cref{alg:find_pareto_cycles} for the pseudocode.
\cref{fig:example_pareto_paths} provides an example instance and the dominance behavior of exemplary bicriteria paths.

\begin{algorithm}[!t]
\caption{Pareto-cycles}
\label{alg:find_pareto_cycles}
\begin{algorithmic}[1]
\Function{Pareto-cycles}{Instance $G_\poly=(G,\ell:E\rightarrow \mathbb{R}_{>0},A: 2^E\rightarrow \mathbb{N}_{0})$, $\alpha >1$}
\State Compute search graph $(G_s,\ell',w)$
\State $n\gets \lvert V(G_s) \rvert$
\For{$u\in V(G_s)$} \Comment{Guess vertex from optimal solution}\label{alg:guess_vertex}
    \State $L_x^0 \gets \emptyset$ for every $x\in V(G_s)\setminus \{u\}$\label{alg:start_subroutine_bellmann}
    \State $L_u^0 \gets \{(0,0)\}$    
    \For{$i=1$ to $n$} \Comment{Bicriteria Bellman-Ford}\label{alg:invariant}
        \For{$x\in V(G_s)$}
            \State $L_x^i \gets L_x^{i-1}$
        \EndFor
        \For{$e=(a,b)\in E(G_s)$}\label{alg:guess_edge}
            \State $L_{b,e}^{i} \gets \{(l_s,w_s) + (\ell'(e), w(e)) \mid (l_s,w_s) \in L_a^{i-1}\}$%
            \State $L_b^i \gets L_b^{i} \oplus L_{b,e}^i$  \Comment{Compute new Pareto set} \label{alg:subroutine_bellmann}
        \EndFor        
    \EndFor
    \State $s_u^* \gets \underset{(l',w')\in L_u^n\setminus{\{(0,0)\}} }{\argmax} \frac{ w'}{(l')^\alpha}$ \Comment{Evaluate all Pareto-cycles $u\rightarrow u$ by $c_\alpha$}\label{alg:find_best}%
\EndFor
\State \Return $\argmax_{(l^*,w^*)\in\{s_u^*\mid u\in V(G_s)\}}\frac{ w^*}{(l^*)^\alpha}$
\EndFunction%
\end{algorithmic}
\end{algorithm}

We note that the bicriteria shortest path problem is usually formulated for strictly positive cost functions. Nevertheless, the correctness proof for the bicriteria Bellman–Ford algorithm extends to our setting, in which some costs may be negative, by essentially the same argument as in the standard case (see \cref{apx:bellman_ford_negative}).
\begin{restatable}{lemma}{BellmanFordLemma}\label{lem:bellman_ford_negative}
    For any start vertex $u\in V(G_s)$, the bicriteria Bellman-Ford subroutine in Algorithm~\ref{alg:find_pareto_cycles}, lines \ref{alg:start_subroutine_bellmann}--\ref{alg:subroutine_bellmann}, computes for each vertex $v\in V(G_s)$ the set of Pareto-optimal $u$-$v$ paths that use at most $n$ edges.
\end{restatable}

It is not difficult to show that a cycle $C^*$ corresponding to a solution maximizing $c_\alpha$ cannot be dominated by any other solution, whether simple or non-simple. 
The proof of the following lemma is deferred to \cref{apx:cycle_dissection}.
\begin{restatable}{lem}{cycleDissection}\label{lem:cycle_dissection}
For any $\alpha>1$, 
    a cycle $C^*$ corresponding to a solution maximizing $c_\alpha$ is Pareto-optimal.
\end{restatable}

The correctness of \cref{alg:find_pareto_cycles} follows almost immediately from 
\cref{lem:bellman_ford_negative} and \cref{lem:cycle_dissection}. Its proof is deferred to \cref{apx:algo_cycles}.

\begin{restatable}{theorem}{algoCycle}\label{thm:algo_cycles}
\cref{alg:find_pareto_cycles} computes, for any given instance $G_\poly$ and $\alpha > 1$, the length and area of the optimal solution for $c_\alpha$. The runtime of the algorithm is $\OO(n^2 \cdot m \cdot \Pm)$ where $\Pm$ is the size of the largest Pareto set during computation.
\end{restatable}
We note that our algorithm does not use the geometric properties of the polygons.
The edge lengths and areas of the polygons could also be defined as arbitrary positive values and do not need to correspond to Euclidean measures.

In the actual implementation, we do not store the respective paths explicitly, but only their current bicriteria costs. 
We note that the exact path of the optimal solution can be reconstructed by saving for each solution a pointer to the previous vertex from which the current solution was derived.

Instead of bounding the size of each Pareto set in each vertex by $\Pm$, we can alternatively bound the runtime by $n\cdot \Am$:
Since we assume that the areas of our polygons are integers and their sum is bounded from above by $\Am$, we can upper bound the weight of any edge in $G_s$ by $\Am$ and lower bound it by $-\Am$. 
As the optimal solution $C^*$ corresponds to a simple cycle, it contains at most $n$ edges.
This means that for each subpath $s_i$ of $C^*$ its weight can be upper bounded by $n \cdot \Am$ and lower bounded by $-n \cdot \Am$. 
Therefore we can upper bound the size of a list $L_u$ containing Pareto-optimal solutions by $2n \cdot \Am + 1$ (we do not save multiple solutions with the same value).

Alternatively, instead of bounding the runtime based on the total area, we can similarly bound the number of different Pareto-optimal solutions by the total perimeter, assuming that each edge length is an integer value. 
Since the edge lengths in the search graph are always positive, the length of any optimal subpath $s_i$ of $C^*$ can always be upper-bounded by $P_{\textrm{total}}$, where $P_{\textrm{total}}$ is the sum of the lengths of all edges of the polygons.
Any solution that exceeds this perimeter threshold can directly be discarded as it can not correspond to some subpath of an optimal solution.
Therefore, the size of a list $L_u^i$ containing Pareto-optimal solutions can also be bounded by $ P_{\textrm{total}}$.
This gives us the following corollary.
\begin{corollary}\label{cor:find_pareto_cycles_pmax}
If all area values are integers, the running time of 
Algorithm~\ref{alg:find_pareto_cycles} is $\OO(n^3\cdot m\cdot \Am)$.
Alternatively, if all edge lengths are integers, the runtime of Algorithm~\ref{alg:find_pareto_cycles} is $\OO(n^2\cdot m\cdot \Pt)$.
\end{corollary}

In contrast to~\cite{PP93}, who upper bounded the intermediate area of any subpath $s_i$ by $\Am$, we used $n\cdot \Am$ instead:
Though the final weight of an optimal solution will always be in $[-\Am,\Am]$, one can construct instances where at least one intermediate path 
always has weight strictly larger than $\Am$ or strictly smaller than $-\Am$. 
In \cref{sec:weight_bounds_instance} we give an example instance for sparsest cut for which any optimal path in the dual graph will have a weight larger than $\Am$ or smaller than $-\Am$ at some point. 
Since in~\cite{PP93} every such solution cannot be found because intermediate paths are always restricted to have weight between $-\Am$ and $\Am$, 
without further arguments one also needs to ensure that paths with associated weight in $[-n\cdot \Am, n\cdot \Am]$ are not removed.
This increases the runtime proposed in~\cite{PP93} to $\OO(n^3\cdot \Am\cdot \log (n\cdot \Am))$ if $n$ Dijkstra iterations are performed.

As mentioned above, the algorithm by Park and Phillips~\cite{PP93} can be adapted directly to obtain a solution that maximizes $c_\alpha$. However, since it is based on finding shortest paths in the search graph $G_s$, its running time grows pseudopolynomially in $\Am$ not only in the worst case but on all instances. The running time of our algorithm can instead be bounded in terms of $\Pm$, which is typically much smaller, though in the worst case it is also pseudopolynomial. In \cref{sec:exp} we present experiments on some real-world inputs showing that this is indeed the case. In particular, we study the effect of the precision of the numbers in the input. For this, we take a real-world instance with large integer area values and
round them to values with fewer significant digits.
This can be seen as reducing the precision of the numbers. While the running time of our algorithm is rather stable and does not increase significantly with the precision of the numbers, the running time of the algorithm by Park and Phillips grows rapidly with the required precision, making it impractical for realistic inputs. More detailed descriptions of the datasets and the observed behavior for both algorithms are given in \cref{sec:exp}.

It has been observed that Pareto sets for realistic inputs of many problems are much smaller than in the worst case. 
A theoretical explanation of this has been given in the framework of smoothed analysis. In this framework, one deviates from the usual worst-case perspective by assuming that an adversarial input is subject to random noise. 
This noise is supposed to model non-adversarial influences that are present in applications. 
Following previous work on smoothed analysis, we assume that an instance of the $\alpha$-circularity problem is given in which the length $\ell(e)$ of each edge $e$ in the graph $G_\poly$ is an independent random variable. 
We assume that the rest of the instance is chosen adversarially and that also the probability distributions of the edge lengths are chosen adversarially. 
To be more specific, we assume that an adversary can choose for each edge $e\in E$ a probability distribution $f_e:[0,1]\to[0,\phi]$ according to which the length $\ell(e)$ of $e$ is chosen independently from the other edge lengths. Here, the parameter $\phi\ge 1$ measures how close the analysis is to a worst-case analysis: For $\phi\to\infty$, the smoothed analysis becomes a worst-case analysis because the adversary can concentrate the entire probability mass on a single value. For $\phi=1$, the smoothed analysis resembles an average-case analysis in which every edge length is chosen uniformly at random from $[0,1]$.
This setup also represents the circularity setup in the geodetic setting, where for example edge lengths are subject to measurement errors and noise.

Using Theorem~1 from \cite{BeierRRV2022}
and the fact that every edge is contained at most $n$ times in any Pareto-optimal solution, it follows that $\Pm$ can be bounded by $\OO(n^2\cdot m^2 \cdot \phi)$ in expectation. 
Combined with \cref{thm:algo_cycles} this yields the following corollary.

\begin{corollary}
   If the length of every edge $e$ from $G_\poly$ is chosen independently according to a probability density $f_e:[0,1]\to[0,\phi]$, the runtime of Algorithm~\ref{alg:find_pareto_cycles} is $\OO(n^4\cdot m^3 \cdot \phi)$ in expectation.
\end{corollary}

\section{Approximation Scheme}\label{sec:FPTAS}
Next, we design an FPTAS for the $\alpha$-circularity problem.
For this purpose we apply a standard rounding scheme to ensure that each iteration of the bicriteria Bellman-Ford algorithm runs in polynomial time in $m$ and $\frac{1}{\varepsilon}$.
We assume that the ratio between the total area $\Am$ and the smallest area $A_{\min}$ is bounded by some polynomial function $f(m)$, i.e., $\frac{\Am}{A_{\min}}\leq f(m)$.

By running Algorithm~\ref{alg:find_pareto_cycles} and replacing the area values $a_i$ of the edges in the search graph with rounded values $a_i'= \left\lfloor \frac{a_i}{K'} \right\rfloor$, where $K'=\frac{\varepsilon \cdot \Am}{f(m)\cdot m}$, we can ensure a polynomial runtime and guarantee that the computed solution is a $(1-\varepsilon)$-approximation to the solution based on the original area values.
The full details and proof are deferred to \cref{apx:FPTAS}.

\begin{restatable}{theorem}{FPTASRounding}\label{lemma:FPTAS}
If $\Am/A_{\min} \leq f(m)$ for some polynomial function $f(m)$, one can compute for any given $\alpha>1$ a $(1-\varepsilon)$-approximation for the $\alpha$-circularity problem in time $\OO\left(\frac{n^3\cdot m^2\cdot f(m)}{\varepsilon}\right)$.
\end{restatable}
\section{\NP{}-hardness of the Circularity Problem}\label{sec:NP-hardness}
In this section we show weak \NP{}-hardness of the $\alpha$-circularity problem for $\alpha \in (1,2]$ via a reduction from the \textit{partition problem}, which is weakly \NP{}-complete:
given a list $Y= (a_1,\ldots, a_n)\in \NN^n$ with $Z:=\sum_{i=1}^n a_i$, we want to find an index set $I\subset [n]$ such that $\sum_{i\in I}a_i = Z/2$. 
For an arbitrary partition instance we construct an $\alpha$-circularity instance $\poly$ such that a specific $\alpha$-circularity score can be achieved if and only if there exists a feasible solution $I\subset [n]$ with $\sum_{i\in I}a_i=Z/2$.

\begin{figure}
\centering
\includegraphics[width=.35\linewidth]{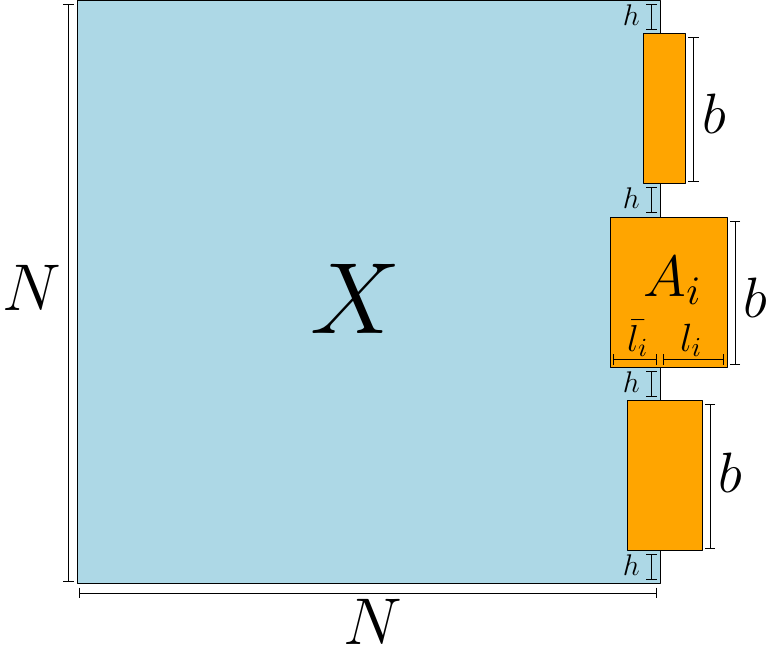}
\caption{Construction visualization for equivalent $\alpha$-circularity instance. }
\label{fig:hardness_instance}
\end{figure}

Intuitively, the construction works as follows: 
We begin with a large axis-aligned square with side length $N$.
Next, we place $n$ axis-aligned rectangles $R_1,\dots, R_n$ along the right border of the square $X$, such that
they partially cover the boundary of $X$.
Additionally, the rectangles are at distance $h$ from each other. 
The value of $N$ will be chosen later.
A visualization of the instance is shown in \cref{fig:hardness_instance}.
The large square, after being overlaid with the rectangles, has remaining area $\Ax$ and a larger perimeter $\Px$. 
We refer to $X$ as the \textit{almost-square}.

Each rectangle $R_i$ has area $a_i$ and including it in a solution containing $X$ increases the total area by $a_i$ and the perimeter by $\frac{a_i}{c}$, for some $c\in \RR_{>0}$. 
For the construction, we assume that $Z\geq 27$, otherwise the partition instance can be solved in constant time.

Each $R_i$ has vertical length $b$ and two horizontal lengths, where $l_i$ extends outside of $X$ and $\bar{l_i}$ extends inside of $X$, see \cref{fig:hardness_instance} for a visualization.
For our construction, the following two conditions must hold:
\begin{align}
    b(l_i+\bar{l_i}) &= a_i, \label{equ:area_cond} \\
    2(l_i-\bar{l_i}) &= \frac{a_i}{c}. \label{equ:perim_cond}
\end{align}
Equation~\eqref{equ:area_cond} corresponds to the area of rectangle $R_i$, which should be equal to $a_i$.
Equation~\eqref{equ:perim_cond} represents the change in perimeter when $R_i$ is added to any solution containing $X$.
Solving both equations for $l_i$ gives us 
\begin{align*}
    l_i &= \frac{a_i}{b} - \bar{l_i} &l_i &= \frac{a_i}{2c} + \bar{l_i}.
\end{align*}
Combining both equations and then solving for $\bar{l_i}$ and then for $l_i$ gives us
\begin{align*}
\bar{l_i} &= \frac{a_i}{2} \left( \frac{1}{b} - \frac{1}{2c} \right), &
l_i &= \frac{a_i}{2} \left( \frac{1}{b} + \frac{1}{2c} \right).
\end{align*}
 \begin{figure}%
    \centering
     \begin{subfigure}[t]{0.3\textwidth} 
         \centering
         \includegraphics[width=\textwidth, page=1]{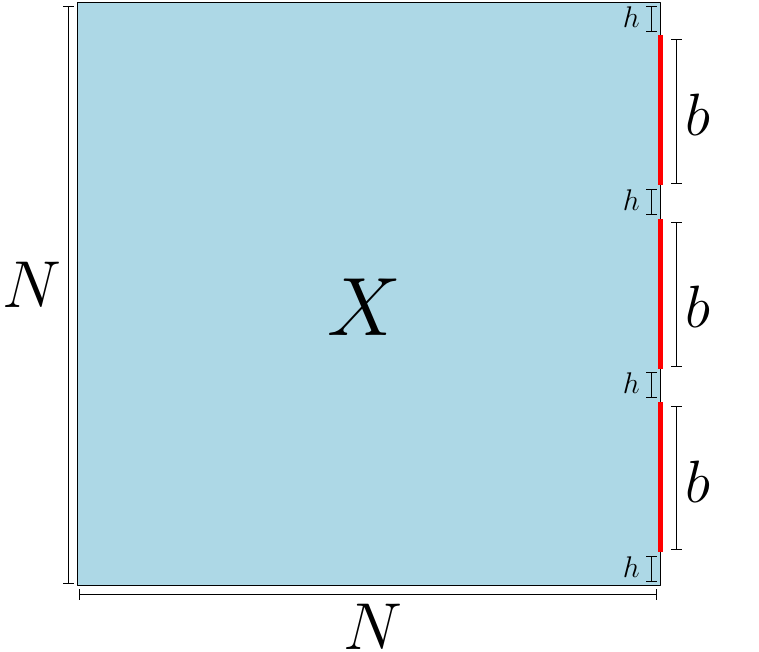} 
         \subcaption{}
         \label{fig:hardness_instance_lines}
     \end{subfigure}
     \hfill
     \begin{subfigure}[t]{0.3\textwidth}
         \centering        
\includegraphics[width=\textwidth, page=2]{figures/hardness_construction.pdf} 
         \subcaption{}
         \label{fig:hardness_instance_smaller}
     \end{subfigure}
    \hfill
    \begin{subfigure}[t]{0.3\textwidth}
        \centering
        \includegraphics[width=\textwidth, page=5]{figures/hardness_construction.pdf} 
        \subcaption{}
        \label{fig:hardness_instance_removing}
    \end{subfigure}
    \caption{Steps for constructing the $\alpha$-circularity hardness instance, as described in Algorithm~\ref{alg:generate_NP_instance}.\\
    \cref{fig:hardness_instance_lines}: Generate square $X$, then remove $n$ line segments of length $b$ at pairwise distance $h$ on the right side. \\\cref{fig:hardness_instance_smaller}: Fit all $n$ rectangles into the right side such that conditions for $\overline{l_i}$ and $l_i$ for each rectangle $R_i$ are fulfilled.\\
    \cref{fig:hardness_instance_removing}: Remove additional area and increase perimeter by clipping away corner parts (decrease in area) or taking detours at the border (increase perimeter).
    }
    \label{fig:construction_hardness}
\end{figure}
We need to show the following conditions to prove that our instance $\poly$ can always be constructed for every possible partition instance $Y$:

\begin{itemize}
    \item[(1)] For every rectangle $R_i$, all side lengths $b,l_i,\bar{l_i}$ are strictly positive. 
    \item[(2)] The rectangles can be placed evenly along the border of $X$ without touching each other, i.e., $h>0$.
    \item[(3)] No rectangle $R_i$ can cut through the almost-square, i.e., for all $i\in [n]$ we have $\bar{l_i} < N$.
\end{itemize}

In the following we assume $b=Z^2$ and $N=Z^5$. 
Additionally, we set $c = \frac{\alpha\cdot \Ax + Z/2\cdot(\alpha -1)}{\Px}$.
We can see that $c>0$ always holds if $\alpha>1$.
Assume that our instance is feasible, so the conditions $(1)$-$(3)$ hold.
We show correctness of the conditions in Lemmata~\ref{lem:construction_1} to \ref{lem:construction_3}.

Since each rectangle $R_i$ can reduce the area of the original square of side length $N$ by at most $a_i$, we know that $\Ax$ is at least $N^2-Z=Z^{10}-Z$.
We can upper bound the inside length $\bar{l_i}$ by $\bar{l_i}=\frac{a_i}{2}(\frac{1}{b}-\frac{1}{2c})\leq \frac{Z}{2}(\frac{1}{Z^2}-0) = \frac{1}{2Z}$ and therefore the total perimeter $\Px$ by $4N+2\sum_{i\in [n]}\bar{l_i}\leq 4Z^5 +2Z\cdot \frac{1}{2Z}= 4Z^5+1$.

\begin{algorithm}[!t]
\caption{Generate $\alpha$-circularity instance}
\label{alg:generate_NP_instance}
\begin{algorithmic}[1]
\Function{$\alpha$-circularity Instance}{Partition Instance $Y=(a_1,\ldots,a_n)\in \mathbb{N}^n, \alpha\in (1,2]$}
\State Set $Z:= \sum_{i \in [n]} a_i$
\State Set $N=Z^5$, $\Ax=Z^{10}-Z$, $\Px=4Z^5+1$, $c = \frac{\alpha \cdot \Ax + Z/2\cdot(\alpha -1)}{\Px},$ \Statex \hspace{1cm} $b=Z^2$ and $h=(N-n\cdot b)/(n+1)$
\State Draw full square with side length $N$
\State Remove $n$ line segments of length $b$ and pairwise distance $h$ on the right side of 
\Statex \hspace{1cm} the full square (\cref{fig:hardness_instance_lines})
\State Place rectangles $R_1,\dots,R_n$ at the removed line segments from the full square such  
\Statex \hspace{1cm} that for every $i\in [n]$ we have $l_i= \frac{a_i}{2}(\frac{1}{b} + \frac{1}{2c})$ and $\bar{l_i} = \frac{a_i}{2} (\frac{1}{b} - \frac{1}{2c})$ (\cref{fig:hardness_instance_smaller})
\State Modify the now almost-square $X$ such that it gains perimeter $1-2\sum_{i\in [n]}\overline{l_i}$ and loses 
\Statex \hspace{1cm} area $Z-b\sum_{i\in [n]}\overline{l_i}$ (\cref{fig:hardness_instance_removing})
\State \Return drawn instance
\EndFunction
\end{algorithmic}
\end{algorithm}

See Algorithm~\ref{alg:generate_NP_instance} for the procedure to generate our instance: Given an arbitrary instance $Y\subset \mathbb{N}^n$, we fix $\Ax=Z^{10}-Z,\Px=4Z^5+1$ and choose $c$ and $h$ accordingly. 
We start with a complete square, remove evenly spaced line segments, and insert our rectangles $R_1,\dots,R_n$ into these gaps. 
Afterwards, the area of the almost-square $X$ is $N^2-b\sum_{i\in [n]}\overline{l_i}$, and its perimeter is $4N + 2\sum_{i\in [n]}\overline{l_i}$.
Assuming $\overline{l_i}>0$ for all $i$, we modify $X$ to increase perimeter by $1-2\sum_{i\in [n]}\overline{l_i}$, which, by our previous reasoning, will correspond to a value of at least $0$,
and decrease area by $Z-b\sum_{i\in [n]}\overline{l_i}$, ensuring that the final area and perimeter satisfy $\Px=4Z^5+1$ and $\Ax=Z^{10}-Z$.
Figure~\ref{fig:construction_hardness} illustrates this process.

From the construction it also follows that area loss and perimeter gain are always non-negative.
By removing corner parts of $X$ (top left in Figure~\ref{fig:hardness_instance_removing}), we can reduce area while keeping the perimeter unchanged. 
By introducing detours inside and outside of $X$ (bottom left in Figure~\ref{fig:hardness_instance_removing}), we can increase its perimeter while maintaining the same area. 

For a partition solution $I\subseteq [n]$ we define $S_I=R_I\cup \{X\}$, where $R_I:=\{R_i \in \{R_1,\ldots,R_n\} \mid i\in I \}$ 
is the set of selected rectangles along with $X$.
This solution yields an $\alpha$-circularity score of 
\begin{align*}
    c_\alpha(S_I) = \frac{\Ax+\sum_{i\in I}A(R_i)}{(\Px+\sum_{i\in I}2(l_i-\bar{l_i}))^\alpha}
    =\frac{\Ax+\sum_{i\in I}(l_i+\bar{l_i})b}{(\Px+\sum_{i\in I}2(l_i-\bar{l_i}))^\alpha}
    = \frac{\Ax+\sum_{i\in I} a_i}{(\Px+\frac{1}{c}\sum_{i\in I}a_i)^\alpha}.
\end{align*} 
By treating $\sum_{i\in I}a_i$ as a continuous variable $x\in \RR_{\geq 0}$, this allows us to express the cost function as $c_\alpha(x)=\frac{\Ax+x}{(\Px+\frac{x}{c})^\alpha}$. 
With $c = \frac{\alpha \cdot \Ax + Z/2\cdot(\alpha -1)}{\Px}$, the function
$c_\alpha$ attains its maximum for $x=Z/2$.
The proof of the following lemma is deferred to \cref{apx:function_maximum_hardness}:
\begin{restatable}{lemma}{FunctionMaximumLemma}\label{lem:function_maximum_hardness}
    The function $c_\alpha(x)$ attains its unique maximum for $x\in \mathbb{R}_{\geq 0}$ at $x=Z/2$.
\end{restatable}
With \cref{lem:function_maximum_hardness} we know that if the instance is feasible, i.e., it can be constructed as previously described, and if the almost-square $X$ is included in the solution, the optimal $\alpha$-circularity is achieved by selecting the rectangles from a partition solution $I\subset [n]$ such that the total area sums up to $Z/2$. 
It remains to show that no better solution exists and that the instance is feasible.
The proofs for the construction Lemmata~\ref{lem:construction_1} to \ref{lem:construction_3}, as well as the condition that no other solution can achieve a better $\alpha$-circularity score than $c_\alpha(S_I)$, are deferred to \cref{apx:NP-hardness}.

\begin{restatable}[Condition (1)]{lemma}{ConstructionFirst}
\label{lem:construction_1}
    For every rectangle $R_i$ we have $b>0$, $l_i>0$ and $\overline{l_i}>0$.
\end{restatable}

\begin{restatable}[Condition (2)]{lemma}{ConstructionSecond}
\label{lem:construction_2}
We can position all rectangles along the border of $X$ such that none overlap and are at distance $h>0$ separated from each other.
\end{restatable}

\begin{restatable}[Condition (3)]{lemma}{ConstructionThird}
\label{lem:construction_3}
No rectangle $R_i$ can cut through the almost-square, i.e., $\bar{l_i}<N$ for every $i\in [n]$.
\end{restatable}

\begin{restatable}{lemma}{SingleRectangle}
\label{lemma:NP-hardness_single_rectangle}
For any nonempty solution $S\subseteq \{R_1,\ldots,R_n,X\}$ we have $c_\alpha(S) \leq \frac{A_X+Z/2}{(P_X+ Z/(2c))^\alpha}$. Additionally, equality holds if and only if $X\in S$ and $\sum_{R_i\in S\setminus \{X\}} a_i = \frac{Z}{2}$.
\end{restatable}

Now, the correctness of the reduction follows directly from the properties established earlier. The proof of the following theorem is deferred to \cref{apx:NP-hardness}.

\begin{restatable}{theorem}{NPhardness}\label{thm:NP-hardness-circ}
    The $\alpha$-circularity problem is weakly \NP{}-hard for $\alpha\in (1,2]$.
\end{restatable}

Using the previously defined construction, one can also show \NP{}-hardness for related extensions of the $\alpha$-circularity problem as well.
In \cref{sec:NP-hardness-districting}, we show that the problem of partitioning all given polygons into exactly $k\in \NN$ classes such that the minimum $\alpha$-circularity score over the union of polygons in each class is maximized is also weakly \NP{}-hard.

\bibliography{refs}

\newpage
\appendix

\section{Omitted Proofs from \cref{sec:Preliminaries}}\label{apx:singleCyclex}

\subsection{Proof of Lemma~\ref{lemma:SingleCycle_x}}
\label{apx:singlecycleproof}
\SingleCyclex*

\begin{proof}
Let $A_1:=A(S_1),A_2:=A(S_2)$ denote the areas of both sets and let $P_1:=\ell(S_1),P_2:=\ell(S_2)$ be their perimeters.
Additionally, we have $P_1,P_2>0$ because each edge of a polygon has a positive length. 
Since the area and perimeter of solution $S_1$ remain unchanged when $S_2$ is removed (and vice versa), we can express the area and perimeter of $S$ as the sum of the individual areas and perimeters of $S_1$ and $S_2$, i.e., $A(S) = A(S_1) + A(S_2)$ and $P(S) = P(S_1) + P(S_2)$.
Assume for contradiction that the lemma does not hold.
This implies
\begin{alignat}{2}
&~c_\alpha(S)= \frac{A_1+A_2}{(P_1+P_2)^\alpha} &&\geq \max\left\{\frac{A_1}{P_1^\alpha},\frac{A_2}{P_2^\alpha}\right\} \geq \frac{A_1}{P_1^\alpha} \nonumber \\
\Leftrightarrow &  ~P_1^\alpha(A_1 + A_2) &&\geq A_1(P_1 + P_2)^\alpha. \tag{i}\label{eq:ineq1}
\end{alignat}
By identical reasoning we also get
\begin{alignat}{2}
&~c_\alpha(S)= \frac{A_1+A_2}{(P_1+P_2)^\alpha} &&\geq \max\left\{\frac{A_1}{P_1^\alpha},\frac{A_2}{P_2^\alpha}\right\} \geq \frac{A_2}{P_2^\alpha} \nonumber \\
\Leftrightarrow &  ~P_2^\alpha(A_1 + A_2) &&\geq A_2(P_1 + P_2)^\alpha. \tag{ii} \label{eq:ineq2}
\end{alignat}
Combining inequalities (\ref{eq:ineq1}) and (\ref{eq:ineq2}) results in
\begin{alignat}{2}
&~(P_1^\alpha + P_2^\alpha)(A_1 + A_2) &&\geq (P_1+P_2)^\alpha(A_1+A_2) \nonumber \\
\Leftrightarrow &~ P_1^\alpha + P_2^\alpha && \geq (P_1 + P_2)^\alpha. \tag{iii} \label{eq:ineq3}
\end{alignat}
We know that $f(x)=x^\alpha$ for $x \geq 0$ and $\alpha>1$ is a strictly convex function. 
Since $f(0)=0^\alpha = 0$, $f$ is also superadditive, so it holds that $f(x_1+x_2)>f(x_1) + f(x_2)$. 
Therefore, since $P_1,P_2>0$, we have $(P_1 + P_2)^\alpha = f(P_1 + P_2) >  f(P_1) + f(P_2) = P_1^\alpha + P_2^\alpha$, which is a contradiction to inequality (\ref{eq:ineq3}).
\end{proof}

\subsection{Proof of Lemma~\ref{lemma:PolyBound}}
\label{apx:PolyBound}
\PolyBound*

\begin{proof}
    We only keep degree-2 vertices required for subdividing potential multi-edges.
After applying all possible contractions, the number of vertices and edges in $G_\poly$ is linear in the number of polygons, i.e., we have $n=\OO(\vert \poly \vert)$ and $m=\OO(\vert \poly \vert)$.
This is because, if degree-2 vertices would be contracted as well, assuming $\vert \poly \vert\geq 2$, every vertex will have degree at least 3 afterwards.
Applying Euler's formula $n-m+\vert \poly \vert = 1$, combined with $2m=\sum_{v\in V} \deg(v)\geq 3n$ yields $n\leq 2\vert\poly\vert-2$ and $m\leq 3\vert \poly \vert-3$.
If there are multiple edges that have the same vertices as endpoints we can add a single vertex in the middle to make the graph simple again.
This will 
increase the number of vertices by at most the number of edges in the multigraph, 
resulting in $n\leq 5\vert \poly \vert -5$.
Therefore, for the final simple graph we have $n=\OO(\vert \poly \vert)$ and $m=\OO(\vert \poly \vert)$.
\end{proof}

\section{Omitted Proofs from \cref{sec:Alg}}\label{apx:Alg}

\subsection{Proof of \cref{lem:weight_enclosed}}\label{apx:weight_enclosed}

\begin{figure}
    \centering
    \begin{minipage}[t]{0.475\textwidth}
        \centering
        \includegraphics[width=\textwidth, page=2]{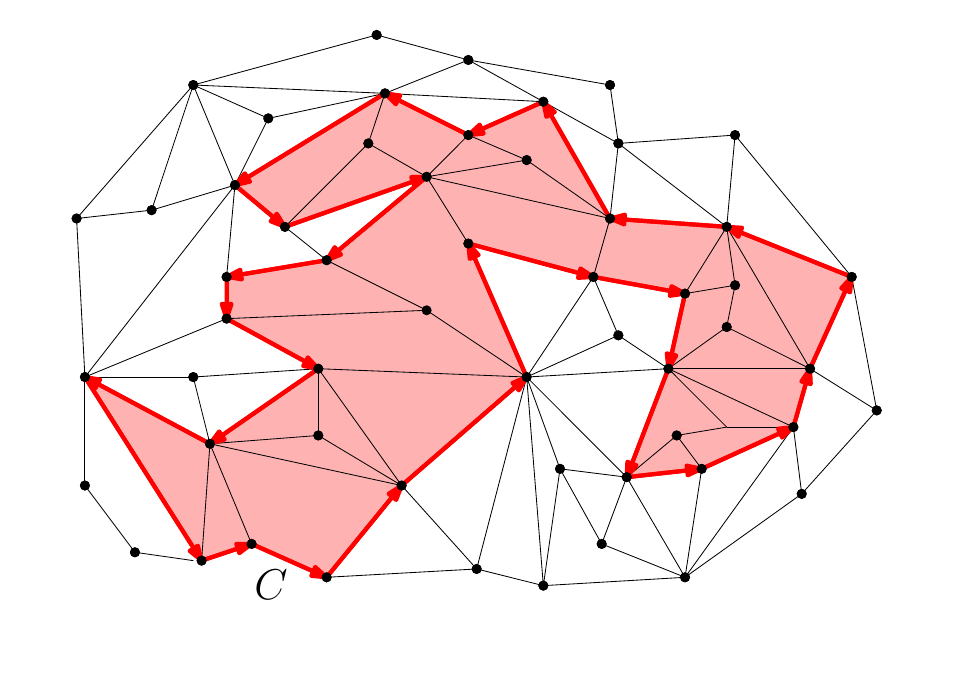} 
        \caption{A possible cycle, drawn in red, in an example instance. The edges are oriented counterclockwise.}
    \end{minipage}
    \hfill
    \begin{minipage}[t]{0.475\textwidth}
        \centering
        \includegraphics[width=\textwidth, page=4]{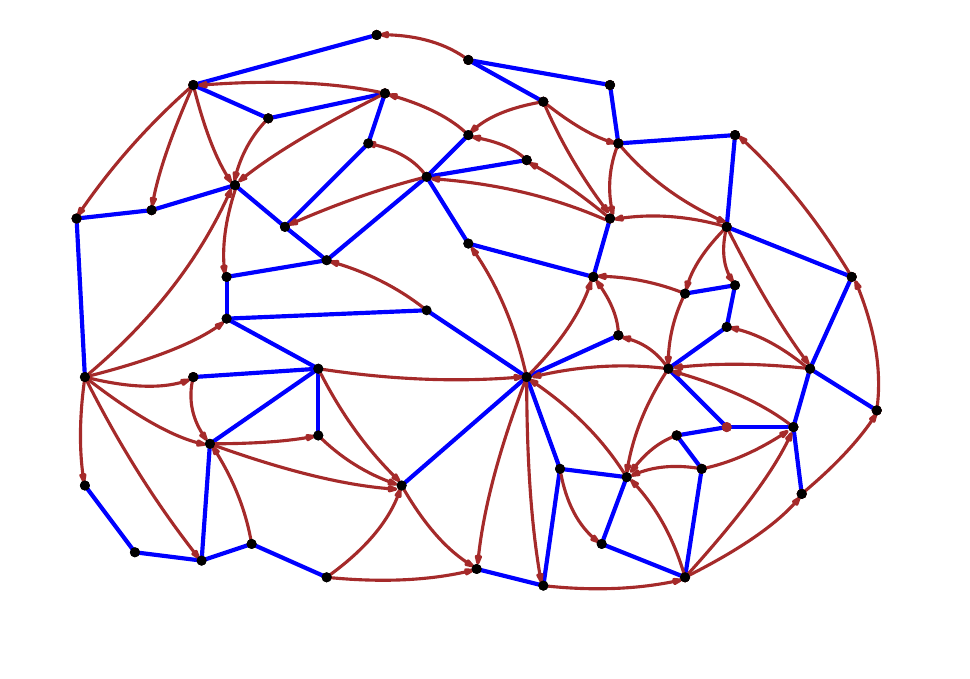} 
        \caption{Arbitrary spanning tree, shown in blue, and all positive edges from the search graph construction, shown in brown.}
        \label{fig:positive_edges}
    \end{minipage}
\end{figure}

\begin{figure}[ht]
    \centering
    \begin{minipage}[t]{0.475\textwidth}
        \centering
        \includegraphics[width=\textwidth, page=13]{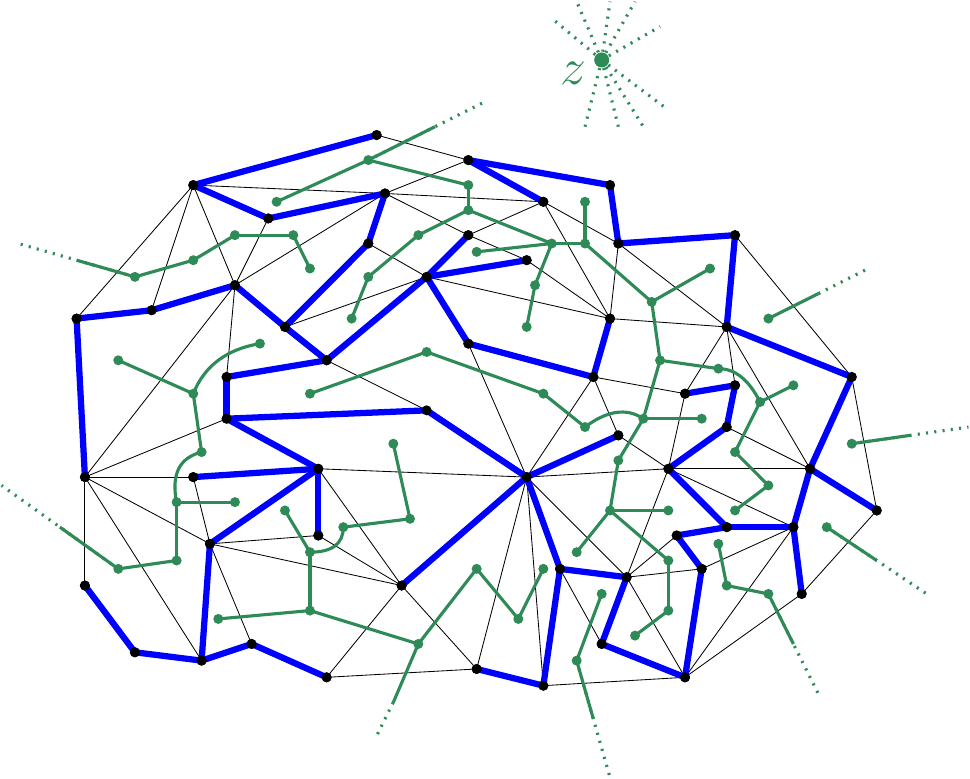}
        \caption{Spanning tree, shown in green, in the dual graph.
        }
        \label{fig:dual_tree}
    \end{minipage}
    \hfill
    \begin{minipage}[t]{0.475\textwidth}
        \centering
        \includegraphics[width=\textwidth, page=18]{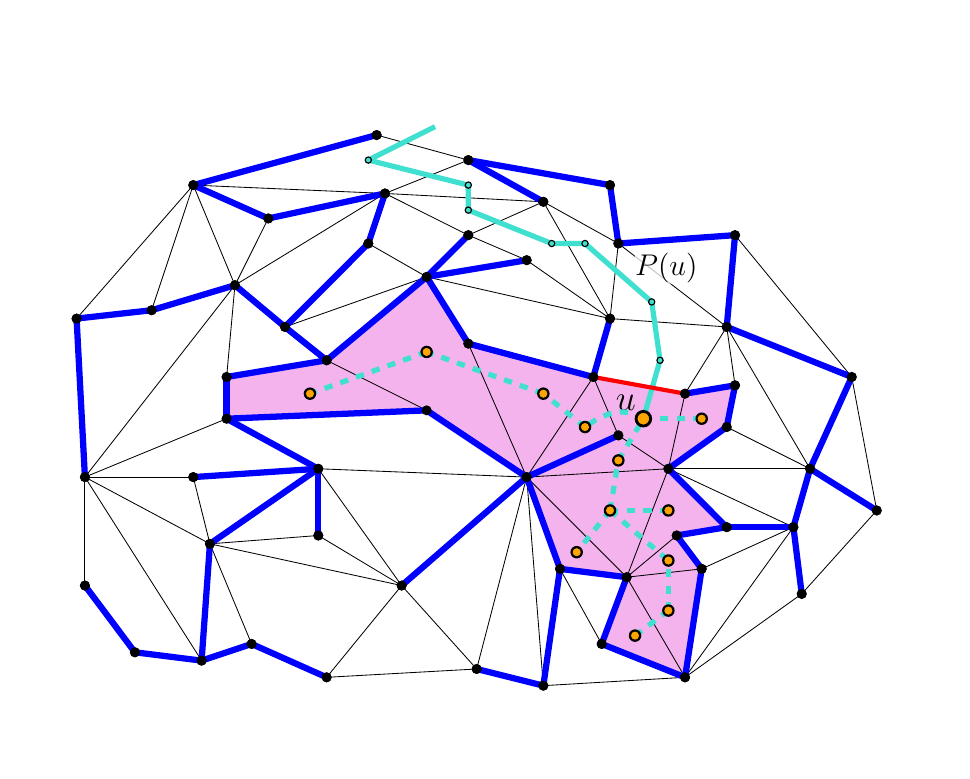}
        \caption{An arbitrary vertex $u$, along with its unique path $P(u)$, colored in cyan, and set of vertices that are enclosed by the unique cycle formed with the spanning tree and the final edge of $E(P(u))$, colored in orange.}
        \label{fig:def_path_enclosed_vertices}
    \end{minipage}
\end{figure}

\begin{figure}
    \centering
        \includegraphics[width=.5\textwidth, page=17]{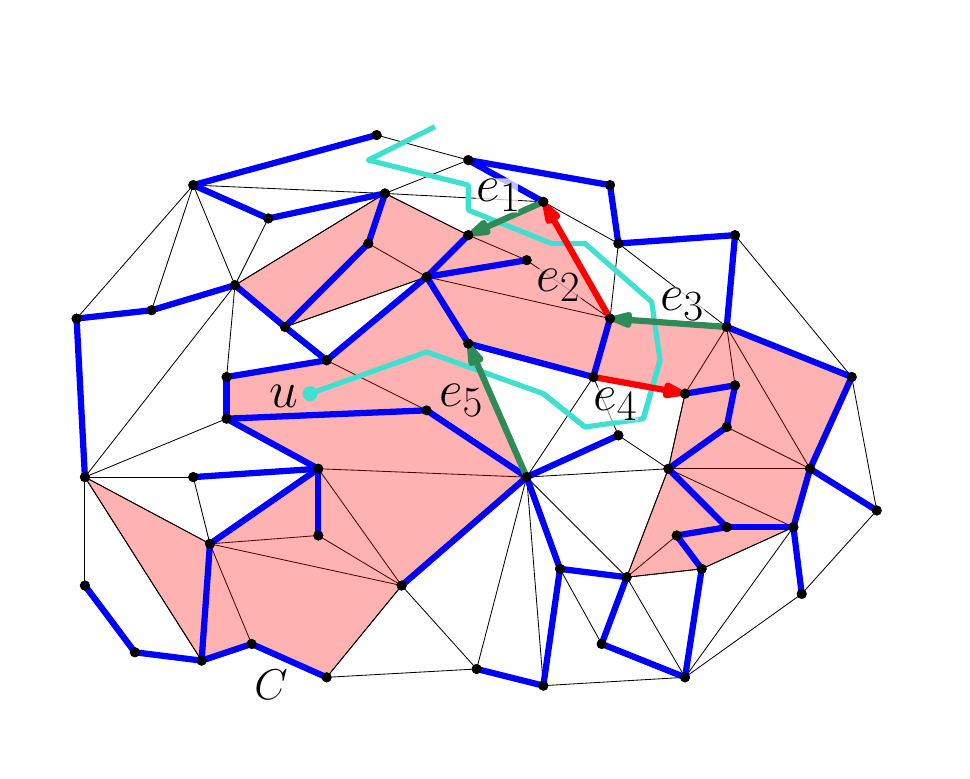}
        \caption{An arbitrary vertex $u$, located in the interior of $C$, along with its unique path $P(u)$. Edges that do not belong to $T$ and are crossed by this path are colored in red or green, depending on the sign of the weight.}
        \label{fig:unique_path_crossed_edges}
\end{figure}

\CycleEnclosingLemma*

\begin{proof}
    Let $G_{\mathcal{T}}=(G, \ell, A)$ for $G=(V,E)$ be an arbitrary connected plane polygon graph. 
    Setting $\ell'=\ell$ trivially preserves the lengths of any simple cycle, so we focus on the correctness of $w$.
    We fix the plane embedding of $G$ and define its corresponding dual graph $G_D=(V_D,E_D)$. 
    Let $c:V_D \rightarrow \RR_{\geq 0}$ be a cost function defined for all vertices of the dual graph, i.e., all faces of $G$ in its embedding, that associates every bounded face with its area and assigns cost $0$ to the outer face. 
    Finally, we define $V_{\text{int}}(C)\subseteq V_D$ as the set of dual vertices enclosed by some simple cycle $C\subseteq E$. 
    We shorten notation and write $c(V')=\sum_{v\in V'}c(v)$ for any set of vertices $V'\subseteq V_D$.
    If $C\subseteq E$ is a simple cycle we have $A(C)=c(V_{\text{int}}(C))$.

    In the following, we describe the construction in more detail. 
    We replace the original undirected graph $G$ with a directed version $G_s=(V_s,E_s)$, where $V_s = V$ and for every edge $e=\{u,v\}\in E$ we include two directed edges $(u,v)$ and $(v,u)$ in $E_s$. 

    For an undirected edge $e=\{u,v\}\in E$ we use $e^D\in E_D$ to denote its unique dual edge and for a directed edge $e=(u,v)\in E_s$ we write $\overline{e}=(v,u)$ to denote its reverse edge and $\hat{e}=\{u,v\}\in E$ to denote the corresponding undirected edge.
     Given a directed edge $e=(u,v)$ we write $\hat{e}^D$ to denote the dual edge of the corresponding undirected edge $\hat{e}=\{u,v\}$.
    
    To define the weight function $w:E_s \rightarrow \RR$, we begin with an arbitrary spanning tree $T\subseteq E$ and set $w(u,v)=w(v,u)=0$ for every edge $(u,v)\in E_s$ with $\{u,v\}\in T$.
    Now, consider any remaining edge $\{u,v\}\in E\setminus T$. 
    Adding either $e=(u,v)$ or $\overline{e}$ to $T$ forms a unique directed cycle $C$ in $G_s$, one oriented clockwise and the other counterclockwise. 
    Assume that including $e$ results in the counterclockwise cycle. 
    We define $w(e)$ as the total cost of the vertices that are enclosed by this cycle, i.e., $w(e)=c(V_{\text{int}}(C))$. 
    For $\overline{e}$ we define $w(\overline{e})=-w(e)$. 
    Intuitively, if the enclosed area of a non-tree edge $e$ lies to its left, we associate its weight with the positive enclosed area, otherwise with the negative area.
    For an illustration of the positively oriented edges in an arbitrary instance with a given spanning tree, see
    \cref{fig:positive_edges}.
    We define $w(E')=\sum_{e\in E'}w(e)$ for a set of edges $E'\subseteq E_s$.

    We use the fact that for any spanning tree $T\subseteq E$ in a planar graph, the set of edges $T^*\subseteq E_D$, which are dual to the complement of $T$, also forms a spanning tree. 
    Let $z\in V_D$ be the dual vertex corresponding to the outer face.
    For each vertex $u\in V_D$, there exists a unique simple path in $(V_D, T^*)$ connecting $u$ to $z$, see \cref{fig:dual_tree} for a visualization.
    In the following we define $P(u)$ as the unique simple path from $z$ to $u$ for $u\neq z$, i.e., $P(u)=(u_1,...,u_b)$ with $u_1=z$ and $u_b=u$. 
    Given some path $P(u)$, $E(P(u))$ denotes the corresponding set of edges, i.e., $E(P(u))=\{\{u_1,u_2\},\ldots, \{u_{b-1},u_b\}\}$, see \cref{fig:def_path_enclosed_vertices} for an example.

    Consider any dual edge $e^D\in E(P(u))$, where $e\in E\setminus T$ is the corresponding primal edge.
    Since $(V_D,T^*)$ is a tree, removing $e^D$ separates $T^*$ into exactly two connected components.
    As $e^D$ is on the unique path from $z$ to $u$, one component must contain $z$, while the other contains $u$.
    Since $e\in E\setminus T$, this edge defines a unique simple cycle in the graph $(V,T\cup \{e\})$, with $e$ being the only non-tree edge.
    This implies that $e^D$ is the only edge in $T^*$ that crosses this cycle. 
    Consequently, the connected component containing $u$ corresponds exactly to the dual vertices enclosed by the unique simple cycle in $(V,T\cup \{e\})$, including $u$.
    Thus, the absolute weight of either directed version of $e$ must equal the total cost of the vertices in this component and therefore also contains $c(u)$ as a summand.

    Now fix any simple directed cycle $C$ in $G_s$.
    We want to show that if $C$ is oriented counterclockwise, then the sum of the edge weights along the cycle equals the total cost of the vertices enclosed by $C$, i.e., $w(C)=c(V_{\text{int}}(C))$.
    Since traversing any simple cycle in the opposite order will by definition result in weight $-w(C)$, we prove this claim for counterclockwise-oriented cycles only.
    
    Let $V_C = V_{\text{int}}(C)\subseteq V_D$ be the set of vertices in the dual graph that are contained in the interior of $C$. 
    We now show that for each vertex $u\in V_C$, its cost $c(u)$ is counted exactly once in the sum $w(C)=\sum_{e\in C}w(e)$, while for any vertex $u\in V_D\setminus V_C$, the cost $c(u)$ does not contribute to the sum or evaluates to zero.

    Let $u\in V_D$ be an arbitrary vertex, possibly in the interior of $C$.  
    Since the outer face vertex $z$ has cost $0$ and is never enclosed by any cycle, we can assume $u\neq z$.
    Let $E_C= \{e\in C \mid \hat{e}^D \in E(P(u))\} = \{e_1,e_2,\ldots,e_k\}$ be the set of edges from $C$, 
    such that for each of their undirected variants the respective dual edge is crossed if we follow the unique path $P(u)$.
    Additionally, for every $i\in [k]$ we assume that $\hat{e_i}^D$ is the $i$-th edge that we encounter when traversing along $P(u)$.
    An example of this is shown in \cref{fig:unique_path_crossed_edges}.
    
    From the Jordan curve theorem we know that the cycle $C$ separates the graph into some interior and exterior region. 
    We can further show that the weight $w(e)$ of some edge $e\in E_C$ is positive if and only if, 
    when traversing the path $P(u)$, whenever $P(u)$ crosses $C$ at some dual edge $\hat{e_i}^D$ with $e_i\in E_C$, we transition from being outside of $C$ to being inside of $C$. 

    \begin{figure}
    \centering
    \begin{minipage}[t]{0.475\textwidth}
        \centering
        \includegraphics[width=\textwidth, page=2]{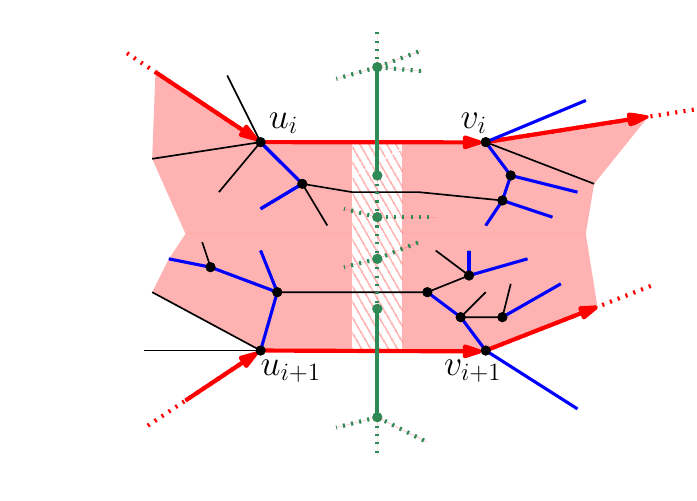}
        \caption{Exemplary situation where orientation of edges does not change after crossing two consecutive edges of $C$ on path to $u$. The subpath in $T^*$ divides $C$ into two regions.}
        \label{fig:edge_orientation}
    \end{minipage}
    \hfill
    \begin{minipage}[t]{0.475\textwidth}
        \centering
       \includegraphics[width=\textwidth, page=3]{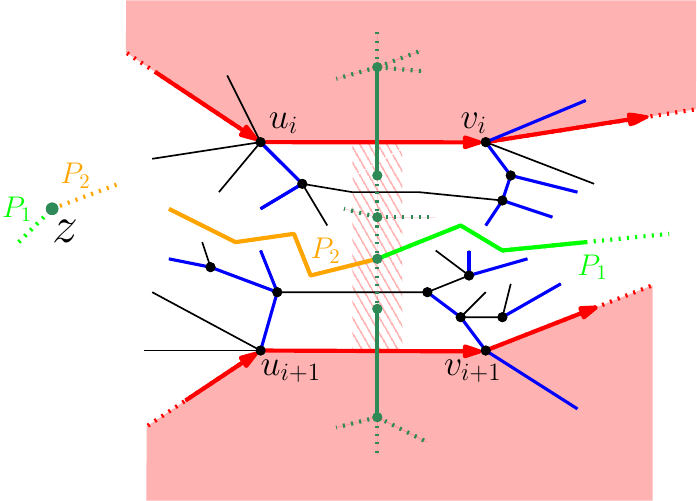}
        \caption{Passing through an outside region, assuming both sides can be connected to $z$, must result in two unconnected regions from $C$.}
        \label{fig:edge_orientation_outside}
    \end{minipage}
\end{figure}
    
    Following along $C$, we can associate edges from $E_s\setminus T$ that are oriented counterclockwise as edges which, when crossed from the outside, correspond to edges that allow us to move from the exterior of $C$ to its interior. 
    Similarly, if the edge is oriented clockwise, we associate it with an edge that, if crossed from the interior of $C$, allows us to move to the exterior.
    
    Since $P(u)$ starts at the outer-face $z$, the first crossing $\hat{e_1}^D$ of $C$ along $P(u)$ leads from the exterior of $C$ to its interior.
    Removing $\hat{e_1}^D$ from $T^*$ gives us one component containing $z$ and one that corresponds to the dual vertices enclosed by the cycle induced by $\hat{e_1}$.
    The region entered at the first crossing therefore corresponds to the region enclosed by this cycle.
    As $C$ is oriented counterclockwise, its interior is located left of $e_1$.
    Since the direction of the relative enclosed region match in both cases, this means that the directed cycle induced by $e_1$ for $(V,T\cup \{\hat{e_1}\})$ is also oriented counterclockwise and we have $w(e_1)>0$.
    
    In the first case we assume that $u$ belongs to the interior of $C$, i.e., $u\in V_C$. 
    We observe that the number of edges crossed must correspond to some odd value $k=2h+1$ for $h\in \NN_0$, because each time we cross an edge of $C$, we change from the exterior of $C$ to the interior or vice versa. 
    Since $P(u)$ starts at $z$ and $u$ lies in the interior of $C$, we must cross an odd number of edges.

    Next, we argue that the relative orientations of two consecutive edges $e_i,e_{i+1}\in E_C$ with respect to the subpath of $P(u)$ between them must be opposite.
    Assume that this were not the case. 
    Let $e_i=(u_i,v_i)$ and $e_{i+1}=(u_{i+1}, v_{i+1})$ be the first two edges crossed along the path $P(u)$ for which the signs do not swap, i.e., the orientation of the edges stays the same.
    
    In the first case, assume the region between both edges is inside of $C$. 
    We can observe that the subpath of $P(u)$, that starts at the first edge crossing $e_i$ and ends at the edge crossing $e_{i+1}$ must split the interior of $C$ into two separate regions, as indicated in \cref{fig:edge_orientation}. 
    Since $e_i$ and $e_{i+1}$ have the same orientation, we can conclude that, if we split the interior along the previously mentioned path, one side must contain both $u_i$ and $u_{i+1}$, while the other side contains $v_i$ and $v_{i+1}$. 
    As all of these points are part of the boundary, it must be possible to connect these respective endpoints by only traversing the interior of $C$.

    In the second case, assume the region between both edges is outside of $C$. 
    Similar to the previous case, we can conclude that the respective subpath of $P(u)$ must divide this region such that one part is still outside while the other part is enclosed by the subpath and $C$. 
    If this were not the case, this would imply that there are paths $P_1,P_2$ which connect either side to $z$ without passing through $C$. 
    But this would mean that if we combine $P_1$ and $P_2$, we would have a cycle in $G_D$, but points from the interior of $C$ are on different sides of this cycle.
    Since this would imply that $C$ cannot be connected, this is a contradiction. 
    For a visualization, see \cref{fig:edge_orientation_outside}. 
    Therefore, the two regions separated by the subpath of $P(u)$ cannot be connected, which means it must be possible to connect the respective endpoints $u_i$ with $u_{i+1}$ and $v_i$ with $v_{i+1}$ via some edge, without violating planarity with respect to either the interior or the exterior region.

    \begin{figure}
    \centering
    \begin{minipage}[t]{0.475\textwidth}
        \centering
    \includegraphics[width=.65\linewidth, page=5]{figures/edge_orientation_swap.pdf}
    \caption{Visualization of the resulting instance and connectivity of the assumed cycle. After removing unused edges, adding point $g$ and including its edges, we have a K5 graph after contraction.}
    \label{fig:K5}
    \end{minipage}
    \hfill
    \begin{minipage}[t]{0.475\textwidth}
        \centering
       \includegraphics[width=.8\textwidth]{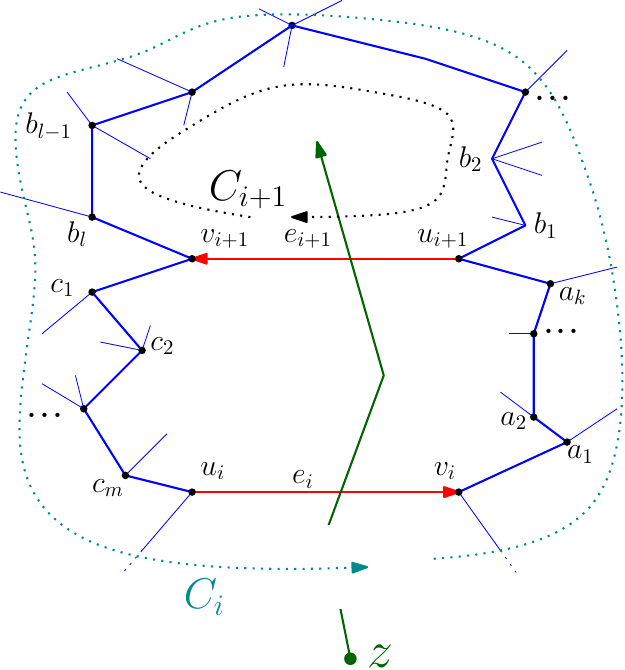}
        \caption{Visualization of directional cycles in proof construction. Edges $e_i$ and $e_{i+1}$ induce respective cycles with spanning tree, shown in blue, that traverse in opposite directions.}
        \label{fig:enclosing_cycles_proof}
    \end{minipage}
\end{figure}

    Next, we only keep in $G$ the undirected versions of the edges of $C$, i.e., the set $\{\hat{e} \in E \mid e\in C\}$. 
    Afterwards, using our previous argument, 
    regardless of if we split up the interior or the exterior region, we can add two edges $e_u=\{u_i,u_{i+1}\}$ and $e_v=\{v_i,v_{i+1}\}$ while preserving planarity.
    Both edges $e_u,e_v$, together with $e_i,e_{i+1}$, must form an enclosed region that is inside or outside of $C$, depending on which case we consider.
    We can place a single vertex $g$ in this region and connect it to all endpoints of the previous four edges, while still maintaining planarity. 
    Next, we contract every edge from $C$ that is not $e_i$ or $e_{i+1}$. 
    By the order of edges it follows that this will result in an edge that connects $v_i$ and $u_{i+1}$. 
    Analogously, there must be an edge which connects $v_{i+1}$ and $u_i$.
    See \cref{fig:K5} for a visualization of the graph construction.
    We can see that our graph must now correspond to $K_5$, which is not planar. 
    Therefore, our original graph $G_s$ could not have been planar which is a contradiction. 

    This means for $E_C$ that the relative orientation of its edges $e_i$ and $e_{i+1}$ must always swap. 
    It remains to show that this change in relative orientation implies that the orientation of the unique simple cycles that include $e_i$ or $e_{i+1}$ are opposite, so the sign of $w(e_i)$ and $w(e_{i+1})$ must be different.
    By our previous arguments we know that $u_i$ and $v_{i+1}$ are in the same relative region, if we use the respective subpath of $P(u)$ as a separator between the regions formed by $e_i$ and $e_{i+1}$.
    The same holds for $v_i$ and $u_{i+1}$. 
    Let $C_i,C_{i+1}$ be the respective two directed cycles formed with $T$.
    From our earlier reasoning, we know that the larger cycle $C_i$ must pass through $u_{i+1}$ after $v_i$.
    Cycle $C_i$ is the larger cycle and must be of the form $C_i=(u_i,v_i,a_1,a_2,\dots,a_k,u_{i+1},b_1,b_2,\dots,b_{l-1},b_l,v_{i+1},c_1,c_2,\dots,c_m,u_i)$. 
    Using $C_i$, we can express the second cycle as $C_{i+1}=(u_{i+1},v_{i+1},b_l,b_{l-1},\dots,b_2,b_1,u_{i+1})$.
    For a visualization, see \cref{fig:enclosing_cycles_proof}.
    Since the direction of the cycle is reversed, we conclude that the signs of $w(e_i)$ and $w(e_{i+1})$ must be different as well. 
    
    Overall, this means that the sign of adjacent edges of $E_C$ must swap at every step.   
    We have already discussed that any edge in $E_C$ will contain either $c(u)$ or $-c(u)$ in its summation.    
    Therefore, we observe that every odd edge $e_{2j+1}\in E_C$ for $j\in \NN_0$ will add the value of $c(u)$ to the cost of $w(C)$, while every even edge $e_{2j} \in E_C$ will subtract the value of $c(u)$. 
    Since the number of edges in $E_C$ is odd, we conclude that $c(u)$ is added $h+1$ times and subtracted $h$ times, so $c(u)$ will be counted exactly once in total.
    Similarly, if $u$ is in the exterior, we either have $E_C=\emptyset$, in which case $c(u)$ does not occur in the weight of any edge of $C$, or we have $E_C=\{e_1,\dots,e_k\}$ for $k=2h$ and $h\in \NN$.
    In the latter case $c(u)$ is added and subtracted an equal amount of times, so all terms including $c(u)$ cancel out and do not contribute to $w(C)$.

    Since this argument holds for every dual vertex $u\in V_D$, we can conclude that $w(C)$ contains the cost of each vertex $u\in V_C$ exactly once, and for every vertex $u\in V_D\setminus V_C$ its total contribution to $w(C)$ is zero. 
    Therefore, we must have $w(C)=c(V_C)$, as required.
\end{proof}

\subsection{Proof of \cref{lem:bellman_ford_negative}}\label{apx:bellman_ford_negative}

\BellmanFordLemma*

\begin{proof}
    We show by induction on $i$ that each time we check the break condition $i=n$ in line~\ref{alg:invariant}, for any $v\in V(G_s)$, the set $L_v^{i-1}$ contains every Pareto-optimal cost vector induced by a $u$-$v$ path using at most $i-1$ edges.
    For simplicity, we will show correctness as if each set contains the actual paths instead of just their respective costs. 
    We will only argue using the costs of these respective paths.
    Since for each tuple in every list there exists a path with the respective cost, this also shows correctness for the case that every list only contains the costs of every solution.

    For $i=1$ there is only one path that uses $0$ edges, which is the empty path. Therefore $L_u^0=\{(0,0)\}$ and $L_x^0=\emptyset$ for every $x\in V(G_s)\setminus \{u\}$ is correct.

    Assume that in a subsequent iteration we reach line~\ref{alg:invariant} and update $i$ to $i+1$. 
    Let $p=(u_1,...,u_j)$ with $u_1=u$ be an arbitrary Pareto-optimal path that uses at most $i$ edges, i.e., $j\leq i+1$.

    If $p$ uses at most $i-1$ edges, by our induction hypothesis it follows that $p$ is contained in $L_{u_j}^{i-1}$ and will be inserted into $L_{u_j}^i$ at the beginning.
    Since $p$ is Pareto-optimal, it will not be removed by any other solution that starts in $u$ and ends in $u_{j}$, including the solutions that contain at most $i$ edges. 
    Therefore $p$ will still be contained in $L_{u_j}^i$ when we reach line~\ref{alg:invariant} again. 

    Assume now that $p$ is a Pareto-optimal solution that uses exactly $i$ edges. 
    Let $p'=(u_1,...,u_{j-1})$ be the subpath that does not contain the last edge of $p$. 
    We observe that $p'$ must be a Pareto-optimal path that uses at most $i-1$ edges. 
    If this were not the case, then there would exist another path $p^*=(v_1,...,v_k)$, with $k\leq i-1$, such that $v_1=u_1,v_k=u_{j-1}$, and the bicriteria cost of $p^*$ dominates the bicriteria cost of $p'$. 
    Now, if we add the edge $(v_k,u_{j})$ to the path $p^*$, we obtain a solution that uses at most $i$ edges and dominates the solution $p$, which  contradicts the assumption that $p$ is Pareto-optimal.
    Therefore, $p'$ is also a Pareto-optimal path that contains at most $i-1$ edges.
    By our induction hypothesis we know that in iteration $i$ the set $L_{u_{j-1}}^{i-1}$ must contain all Pareto-optimal paths using at most $i-1$ edges. Hence, solution $p'$ is also contained in $L_{u_{j-1}}^{i-1}$. 
    Since in every iteration we evaluate each edge, the edge $(u_{j-1},u_j)$ will also be considered. As a result, the solution $p$ will be added to the set $L_{u_j}^{i}$ in this step. 
    Since in line~\ref{alg:subroutine_bellmann} we only remove dominated solutions, it follows that solution $p$ will not be removed. Hence, the solution will be present in this set when $i$ is set to $i+1$.    
\end{proof}

\subsection{Proof of \cref{lem:cycle_dissection}}\label{apx:cycle_dissection}

\cycleDissection*

\begin{proof}
    Let $C^*=(u_1^*,\dots,u_k^*,u_1^*)$ be an optimal cycle that maximizes $\frac{A(C^*)}{\ell(C^*)^\alpha}$. 
    For $i \leq k$, we define the partial path $C_i^*=(u_1^*,\dots,u_i^*)$ as the subpath of the optimal cycle up to vertex $u_i^*$.
    Let $s_i=\left(\sum_{j=1}^{i-1}\ell'(u_j^*,u_{j+1}^*), \sum_{j=1}^{i-1}w(u_j^*,u_{j+1}^*)\right)$ be the cost of this partial solution in the search graph $G_s$, and let $s^*=s_k+(\ell(u_k^*, u_1^*), w(u_k^*, u_1^*))$ denote the cost of the final simple cycle. 
    Note that by \cref{lem:weight_enclosed} $l_{s^*}$ corresponds to the perimeter of the optimal cycle $C^*$, and $\vert w_{s^*} \vert$ corresponds exactly to the enclosed area of $C^*$, i.e., $l_{s^*}=\ell(C^*)$ and $\vert w_{s^*} \vert = A(C^*)$.
    
    We assume that $C^*$ is ordered counterclockwise, i.e., $w(C^*)=A(C^*)$.
    Let $S^*\subseteq \poly$ denote the set of polygons enclosed by $C^*$.
    It is straightforward to show that, when comparing this solution with all solutions corresponding to simple cycles, the bicriteria solution $s^*=(\ell(C^*), A(C^*))$ must be Pareto-optimal. 
    If this were not the case, it would imply the existence of a simple cycle solution $s'=(l_{s'}, w_{s'})$ such that $s' \prec s^*$. 
    By Lemma~\ref{lem:weight_enclosed} we know that $s'$ corresponds to a selection of polygons $S'$ such that $\ell(S')=l_{s'}$ and $A(S') = w_{s'}$. We assume that $s'$ is ordered counterclockwise, because if this were not the case it cannot dominate $s^*$ because $w_{s'}$ would be negative. 
    But this also implies $\frac{A(S')}{\ell(S')^\alpha} > \frac{A(S^*)}{\ell(S^*)^\alpha}$, because either the perimeter or the area (or both) of $S'$ is strictly better than $S^*$, while the remaining objective is at least as good. 
    This is a contradiction to $C^*$ being an optimal solution.

    Therefore, assume that there exists some Pareto-optimal cycle $s'=(u_1',\dots, u_r',u_1')$ that is non-simple and $s'\prec s^*$. 
    We can recursively construct multiple simple cycles $s^1,\dots, s^g$ from $s'$ as follows: 
    Let $v\in s'$ be the first vertex, with the exception of $u_1'$, which appears multiple times in $s'$. 
    Then we can write $s' = (u_1', s_{\text{prev}},v, s_{\text{inter}}, v, s_{\text{after}}, u_1')$ where $s_{\text{prev}}, s_{\text{inter}}$ and $s_{\text{after}}$ represent the sequences of vertices before, between, and after the first two appearances of $v$, respectively. 
    Since $G_s$ contains no self-loops, we have $s_{\text{inter}}\neq \emptyset$.
    Therefore, we can construct two new cycles $s^1, s^2$ from $s'$ by setting $s^1=(u_1', s_{\text{prev}},v, s_{\text{after}}, u_1')$ and $s^2=(v,s_{\text{inter}},v)$. 
    Both cycles contain, by the assumption $s_{\text{inter}}\neq \emptyset$, at least three vertices.    
    We repeat this procedure recursively for $s^1,s^2$ until we are left with only simple cycles.
    In case that $u_1'$ appears at least three times in $s'$ we can analogously express $s'$ as $(u_1',s_{\text{prev}},u_1',s_{\text{after}},u_1')$ and decompose it into two cycles $s^1=(u_1',s_{\text{prev}},u_1')$ and $s^2=(u_1',s_{\text{after}},u_1')$ and recurse on both.

    For an arbitrary (possibly non-simple) cycle $s'=(u_1',\dots, u_r',u_1')$, we define $\ell'(s'):=\sum_{i=1}^{r} \ell'\left(u_i', u_{(i\text{ mod r})+1}'\right)$ and $w(s'):=\sum_{i=1}^{r} w\left(u_i', u_{(i\text{ mod r})+1}'\right)$ as its respective length and area.

    Let $C_{\text{mult}}=\{c^1,\dots, c^g\}$ be the set of simple cycles resulting from the previously defined recursive procedure. 
    We assume that each cycle is ordered counterclockwise, since this can only improve the circularity score of $s'$.
    We observe that each edge of $s'$ appears exactly once in one of the simple cycles $c^1,\dots,c^g$. 
    Thus, if the cycles $c^1,\dots,c^g$ are oriented as in $s'$, we have $\ell'(s') = \sum_{j=1}^g \ell'(c^j)$ and $w(s') = \sum_{j=1}^g w(c^j)$.
    We can assume that $w(c^i)\neq 0$ for any $i\in[g]$, since removing any such cycle from $C_{\text{mult}}$ strictly decreases the overall perimeter of $s'$, hence improves the objective score of $s'$.
    Therefore, for every $i\in [g]$, the simple cycle $c^i$ corresponds to a valid selection of polygons with a total area $w(c^i)>0$ and perimeter $\ell'(c^i)>0$.
    Since $s' \prec s^*$, it follows that 
    $w(s^*)/\ell'(s^*)^\alpha < w(s')/\ell'(s')^\alpha$. 
    On the other hand, the solution $s'$ can be written as the sum of individual areas and perimeters of independent polygons (where some polygons could have been selected multiple times).
    One can verify that \cref{lemma:SingleCycle_x} also holds if polygons are selected multiple times, as long as the total area and total perimeter correspond to the sum of the individual areas and perimeters of the selected polygons.
    Therefore it follows that
    there exists some simple cycle solution $c^j\in S'$ such that its circularity score is strictly better than that of solution $s'$, i.e., $w(s')/\ell'(s')^\alpha < \max_{j\in [g]} w(c^j)/\ell'(c^j)^\alpha$. 
    But this also means that $w(s^*)/\ell'(s^*)^\alpha < \max_{j\in [g]} w(c^j)/\ell'(c^j)^\alpha$, which is a contradiction to our assumption that $s^*$ is the solution with the best $\alpha$-circularity score.    
\end{proof}

\subsection{Proof of \cref{thm:algo_cycles}}\label{apx:algo_cycles}

\algoCycle*

\begin{proof}%
By Lemma~\ref{lem:cycle_dissection} we know that an optimal cycle solution $s^*$ is always Pareto-optimal.
By Lemma~\ref{lem:bellman_ford_negative}, the Bellman-Ford algorithm in Algorithm~\ref{alg:find_pareto_cycles} (lines \ref{alg:start_subroutine_bellmann} to \ref{alg:subroutine_bellmann}) computes the set of Pareto-optimal paths for any $u\in V(G_s)$, using at most $n$ edges.
Since an optimal solution $C^*$ can only consist of at most $n$ edges, if $u=u_1^*$ in line~\ref{alg:guess_vertex}, we will always compute $s^*$ after at most $n$ iterations.
Because $s^*$ is optimal, its $\alpha$-circularity score $w(s^*)/\ell(s^*)^\alpha$ must be maximal among all Pareto-optimal solutions. 
As a result, $s^*$ will always be stored as $s_u^*$ in line~\ref{alg:find_best}, and will never be replaced by another solution afterwards.

One can construct the search graph by first computing an arbitrary rooted tree in $\OO(n+m)$, use it to compute the dual complementary tree $G_D$ in $\OO(n+m)$, and at last, starting from the outer face $z$, recursively visit the vertices in the dual graph and set their appropriate weights in a bottom-up approach. As each dual vertex is visited exactly once and each recursion step takes constant time, the overall runtime of $\OO(n+m)$ follows.

In a naive approach, computing $\PP_1 \oplus \PP_2$ takes time $O(\vert \PP_1 \vert \cdot \vert \PP_2 \vert)$ by comparing every solution from $\PP_1$ with every solution from $\PP_2$ and only keeping those that are not dominated by any other solution.
By storing the Pareto sets in a sorted order (e.g., increasing in the lengths of the solutions) we can compute $\PP_1 \oplus \PP_2$ using a sweep procedure in $O(\vert \PP_1\vert+ \vert \PP_2 \vert) = \OO(\Pm)$.
For our polygon instance we can assume $m\geq n$. As we have three for-loops where two iterate over all vertices and one over all edges, we get a final runtime of $\OO(n^2 \cdot m \cdot \Pm)$.
\end{proof}

\section{Omitted Proofs from Section~\ref{sec:FPTAS}}\label{apx:fptas}

\subsection{Proof of \cref{lemma:FPTAS}}\label{apx:FPTAS}

\FPTASRounding*

\begin{proof}
Assume we have edge area costs $a_1,\dots, a_m\in \mathbb{R}$ and edge length costs $l_1,\dots,l_m\in \mathbb{R}_{> 0}$ in the search graph. 
Additionally, let $\alpha>1$ be arbitrarily chosen. 
For a given $\varepsilon > 0$ we wish to scale all areas by some constant $K\in \mathbb{R}$ such that in the scaled instance we have polynomial runtime and the computed optimal solution in the new instance is within a $(1-\varepsilon)$-factor of the original optimal solution.

We note that in the construction of the search graph,  the edge area costs are always upper bounded by $\Am$ and lower bounded by $-\Am$. 
Additionally, each edge area cost is either at least $A_{\text{min}}$ or at most $-A_{\text{min}}$ or equal to zero in case they belong to edges of the tree. 
Thus we have $a_i\in [-\Am,-A_{\text{min}}]\cup [A_{\text{min}},\Am] \cup \{0\}$.

Let $S\subseteq [m]$ be the edges forming some simple cycle $E_S\subseteq E = \{e_1,\ldots, e_m\}$ in our search graph, where $E_S=\{e_i\in E\mid i\in S\}$ encloses a positive area (so the cycle is traversing counterclockwise). 
Assume $e_i=(a_i,l_i)$ for each edge $e_i\in E_S$, so we have area $A(S)=\sum_{i\in S} a_i>0$ and perimeter $P(S)=\sum_{i\in S} l_i>0$. 
The $\alpha$-circularity cost of this solution is $c_\alpha(S)=\frac{A(S)}{P(S)^\alpha}=\frac{\sum_{i\in S} a_i}{\left(\sum_{i\in S} l_i\right)^\alpha}$.

We define the scaled version of the areas as $a_i':=\left\lfloor \frac{a_i}{K} \right\rfloor$. Let $c'_\alpha$ be the cost function using these scaled values, i.e., $c_\alpha'(S)=\frac{\sum_{i\in S} a_i'}{\left(\sum_{i\in S} l_i\right)^\alpha}$.
Since we aim to solve the $\alpha$-circularity problem in the scaled instance, we need to ensure that the $\alpha$-circularity cost does not deviate significantly when translating the computed solution back to the original setting. 
For this matter, we consider the rescaled areas $a_i^*:= \left\lfloor \frac{a_i}{K} \right\rfloor \cdot K$ and the corresponding rescaled $\alpha$-circularity cost $c_\alpha^*(S)=\frac{\sum_{i\in S} a_i^*}{\left(\sum_{i\in S} l_i\right)^\alpha}$.

Let $S_{\text{OPT}}$ be the optimal $\alpha$-circularity solution for the cost function $c_\alpha$, and let $S'$ be the optimal solution for the scaled cost function $c_\alpha'$.

First we notice that an optimal solution for $c_\alpha'$ is also an optimal solution for $c_\alpha^*$. 
Let $S\subseteq [m]$ be an arbitrary feasible solution. 
We have
\begin{align*}
    c_\alpha^*(S)&=\frac{\sum_{i\in S} a_i^*}{\left(\sum_{i\in S} l_i\right)^\alpha}
    = \frac{\sum_{i\in S} a_i'\cdot K}{\left(\sum_{i\in S} l_i\right)^\alpha}
    = K\cdot \frac{\sum_{i\in S} a_i'}{\left(\sum_{i\in S} l_i\right)^\alpha}
    = K\cdot c_\alpha'(S).
\end{align*}
Since every feasible solution is scaled by the same constant, a solution $S$ that maximizes $c_\alpha^*(S)$ must also maximize $c_\alpha'(S)$.

Next, we analyze how much the $\alpha$-circularity cost $c_\alpha$ can change when we use $c_\alpha^*$ instead.
For every $i\in [m]$, the rescaled areas satisfy $a_i^*= \left\lfloor \frac{a_i}{K} \right\rfloor \cdot K \leq \frac{a_i}{K}\cdot K = a_i$ and $a_i^*= \left\lfloor \frac{a_i}{K} \right\rfloor \cdot K > (\frac{a_i}{K}-1)\cdot K = a_i-K$. 
Therefore, for every possible solution $S$ we have $c_\alpha^*(S) = \frac{\sum_{i\in S} a_i^*}{P(S)^\alpha}\leq \frac{\sum_{i\in S} a_i}{P(S)^\alpha} = c_\alpha(S)$ and
$c_\alpha^*(S) = \frac{\sum_{i\in S} a_i^*}{P(S)^\alpha}> \frac{\sum_{i\in S} a_i-K}{P(S)^\alpha} = c_\alpha(S)-\frac{\vert S \vert\cdot K}{P(S)^\alpha}$.

Using our previous results, we now show that the cost of $c_\alpha(S')$ is close to $c_\alpha(S_{\text{OPT}})$:
\begin{align*}
    c_\alpha(S') \geq c_\alpha^*(S') \geq c_\alpha^*(S_{\text{OPT}}) > c_\alpha(S_{\text{OPT}})-\frac{\vert S_{\text{OPT}}\vert\cdot K}{P(S_{\text{OPT}})^\alpha} 
\end{align*}

Finally, we consider the ratio of the optimal solution to the solution in the scaled instance:
\begin{align*}
    \frac{c_\alpha(S')}{c_\alpha(S_{\text{OPT}})} > \frac{c_\alpha(S_{\text{OPT}})-\frac{\vert S_{\text{OPT}}\vert\cdot K}{P(S_{\text{OPT}})^\alpha}}{c_\alpha(S_{\text{OPT}})}
    = 1 - \frac{\left(\frac{\vert S_{\text{OPT}}\vert\cdot K}{P(S_{\text{OPT}})^\alpha}\right)}{\left(\frac{A(S_{\text{OPT}})}{P(S_{\text{OPT}})^\alpha}\right)} = 1 - \frac{\vert S_{\text{OPT}}\vert\cdot K}{A(S_{\text{OPT}})} > 1 - \frac{m\cdot K}{A_{\min}}.
\end{align*}
By choosing $K\leq \frac{\varepsilon\cdot A_{\min}}{m}$, we obtain $1-\frac{m\cdot K}{A_{\min}} \geq 1-\varepsilon$ as required.

For each area value $a_i$ we have $a_i\in [-\Am, \Am]$. In the scaled version we therefore have
\begin{align*}
    a_i'=\left\lfloor \frac{a_i}{K} \right\rfloor \leq \frac{a_i}{K}\leq \frac{\Am}{K}
\end{align*}
and
\begin{align*}
     a_i'=\left\lfloor\frac{a_i}{K}\right\rfloor> \frac{a_i}{K} - 1\geq - \frac{\Am}{K} -1.
\end{align*}
From the previous results, we know that if $K\leq \frac{\varepsilon\cdot A_{\min}}{m}$, we achieve a $1-\varepsilon$ approximation. 
Since $\frac{\Am}{A_{\min}}\leq f(m) \Leftrightarrow A_{\min}\geq \frac{\Am}{f(m)}$ we can use $K'=\frac{\varepsilon\cdot \Am}{f(m)\cdot m}\leq \frac{\varepsilon \cdot A_{\min}}{m}=K$ to achieve a $1-\varepsilon$ approximation.

For the runtime, using $K'$ as our scaling value, we get
\begin{align*}
    a_i' \leq \frac{\Am}{K'} = \frac{f(m)\cdot m}{\varepsilon}
\end{align*}
as well as
\begin{align*}
     a_i' > - \frac{\Am}{K'} -1 =-\frac{f(m) \cdot m}{\varepsilon} -1.
\end{align*}
Since $a_i'\in \mathbb{Z}$, we can bound $a_i'$ by $\frac{2f(m)\cdot m}{\varepsilon}+1$, which is polynomial in $m$ and $\frac{1}{\varepsilon}$.

For our initial runtime bound we have argued that $\Pm$ can be replaced by $n\cdot \Am$ by using $\Am$ as an upper bound on the weight of every edge in the search graph and using the fact that the optimal solution can only consist of at most $n$ edges, so any intermediate optimal solution weight will always be in the interval $[-n\cdot \Am,n\cdot \Am]$.
Since our rescaled edge area weights are in the interval $(-\frac{f(m)\cdot m}{\varepsilon}-1,\frac{f(m)\cdot m}{\varepsilon}]$, the weight of any intermediate optimal solution will be in the interval $(-\frac{f(m)\cdot m\cdot n}{\varepsilon}-n,\frac{f(m)\cdot m\cdot n}{\varepsilon}]$.
Therefore we can bound the size of the new Pareto set by $\OO(\frac{f(m)\cdot m\cdot n}{\varepsilon})$.
Using this bound instead of $\Pm$ gives us a final runtime of $\OO\left(\frac{n^3\cdot m^2\cdot f(m)}{\varepsilon}\right)$ while achieving a $(1-\varepsilon)$-approximation.
\end{proof}

\section{Omitted Proofs from Section~\ref{sec:NP-hardness}}\label{apx:NP-hardness}

\subsection{Proof of \cref{lem:function_maximum_hardness}}\label{apx:function_maximum_hardness}

\FunctionMaximumLemma*

\begin{proof}
For the first derivative of our function $c_\alpha(x)=\frac{\Ax+x}{\left(\Px+\frac{x}{c}\right)^\alpha}$ we get 
\begin{align*}
    c_\alpha'(x)=\frac{(\Px+\frac{x}{c})^\alpha-\frac{\alpha}{c}(\Ax+x)\left(\Px+\frac{x}{c}\right)^{\alpha - 1}}{\left(\Px + \frac{x}{c}\right)^{2\alpha}}.
\end{align*}
For the second derivative we obtain
\begin{align*}
    c_\alpha''(x)=\frac{\alpha\cdot(-2c\Px + (\alpha - 1)x + (\alpha + 1)\Ax)}{\left(\Px+\frac{x}{c}\right)^\alpha \left(c\Px+x\right)^2}.
\end{align*}

Next, we analyze the relevant extreme points of $c_\alpha$, i.e., the values $x\geq 0$ such that $c_\alpha'(x)=0$.
Since the denominator of $c_\alpha'(x)$ is always positive for $x\geq 0$ and $c\geq 0$, we only need to check when the numerator is zero:
\begin{align*}
    &~0 = \left(\Px+\frac{x}{c}\right)^\alpha-\frac{\alpha}{c}(\Ax+x)\left(\Px+\frac{x}{c}\right)^{\alpha - 1}\\
    \Leftrightarrow &~0 = \left(\Px+\frac{x}{c}\right)^{\alpha - 1}\left(\left(\Px+\frac{x}{c}\right) - \frac{\alpha(\Ax+x)}{c}\right).
\end{align*}
The right-hand side of the equation can only be zero if one of the factors is zero. 
Since the first factor can never be zero, we only need to consider the second factor:
\begin{alignat*}{2}
    &\left(\Px+\frac{x}{c}\right)\cdot c &&= \alpha(\Ax+x) \\
    \Leftrightarrow~& \Px \cdot c + x &&=\alpha(\Ax +x) \qquad\qquad |~c = \frac{\alpha\cdot \Ax + Z/2\cdot(\alpha -1)}{\Px}\\
    \Leftrightarrow~& \alpha \Ax + Z/2\cdot(\alpha -1) +x&&= \alpha \Ax + \alpha\cdot x  \\
    \Leftrightarrow~& Z/2 &&= x,
\end{alignat*}
where in the last step we used $(\alpha -1)>0$.
This means that for $x\geq 0$, the only extreme point occurs at $x=Z/2$. 
Substituting $x=Z/2$ and $c = \frac{\alpha\cdot \Ax + Z/2\cdot(\alpha -1)}{\Px}$ into the second derivative $c_\alpha''(x)$ gives us
\begin{align*}
    c_\alpha''\left(\frac{Z}{2}\right) &= \frac{\alpha\cdot\left(-2 \cdot \frac{\alpha\cdot \Ax + Z/2\cdot(\alpha -1)}{\Px} \cdot \Px + (\alpha - 1)\frac{Z}{2} + (\alpha + 1)\Ax\right)}{\left(\Px+\frac{Z/2}{\frac{\alpha\cdot \Ax + Z/2\cdot(\alpha -1)}{\Px}}\right)^\alpha\left(\frac{\alpha\cdot \Ax + Z/2\cdot(\alpha -1)}{\Px}\cdot \Px+\frac{Z}{2}\right)^2} \\
    &= \frac{\alpha\cdot(-2 \cdot \alpha\cdot \Ax - Z\cdot(\alpha -1) + (\alpha - 1)\frac{Z}{2} + (\alpha + 1)\Ax)}{\left(\Px+\frac{Z\cdot \Px}{2(\alpha \Ax + \frac{Z}{2}(\alpha - 1))}\right)^\alpha\left(\alpha \Ax +\frac{Z}{2}(\alpha - 1) + \frac{Z}{2}\right)^2}.
\end{align*}
We can see that the denominator is always positive, since all terms inside it must always be greater than zero. 
On the other hand, for the numerator (ignoring the positive factor $\alpha$ at the beginning), we get:
\begin{align*}
    &~-2\alpha \Ax - Z(\alpha - 1) + (\alpha - 1) \frac{Z}{2} + (\alpha + 1)\Ax \\
    =&~ \frac{1}{2}(\alpha -1)(-2\Ax - Z).
\end{align*}
Since $\alpha > 1, \Ax>0$ and $Z\geq 27$, this expression is always negative. 
Therefore, the one and only extreme point for the function $c_\alpha(x)$ with $x\geq 0$ occurs at $x=\frac{Z}{2}$, and it is a maximum.
\end{proof}

\ConstructionFirst*

\begin{proof}
    $b>0$ follows immediately from $b=Z^2$ and $Z\geq 27$.
    It is also trivial for $l_i=\frac{a_i}{2}(\frac{1}{b}+\frac{1}{2c})$, since $a_i,b,c>0$. 
    Now, we consider $\bar{l_i}=\frac{a_i}{2}(\frac{1}{b}-\frac{1}{2c})$. 
    This value can only be negative or zero if $\frac{1}{b}\leq\frac{1}{2c}$. 
    Using our previously derived values for $c$ and $b$, we get
\begin{alignat*}{2}
    &\frac{1}{Z^2} && \leq \frac{\Px}{2\alpha \Ax+Z(\alpha - 1)} \hspace{2cm} |~\Ax = Z^{10}-Z, \Px=4Z^5+1 \\
    \Leftrightarrow~ &\frac{1}{Z^2} && \leq \frac{4Z^5+1}{2\alpha Z^{10}-2\alpha Z+Z(\alpha - 1)} \\
    \Leftrightarrow~ &2\alpha Z^{10}-2\alpha Z + Z(\alpha -1) && \leq 4Z^7+Z^2 \\
    \Leftrightarrow~ &2\alpha Z^{9} -\alpha && \leq 4Z^6+Z+1\\
    \Leftrightarrow~ &\alpha && \leq \frac{4Z^6+Z+1}{2Z^9-1}.
\end{alignat*}
We can see that for $Z\geq 27$ the right hand side is always strictly less than 1. As we have $\alpha>1$, this inequality can never be fulfilled, thus the original inequality cannot hold.
Therefore, $\bar{l_i}> 0$ is always true.
\end{proof}

\ConstructionSecond*

\begin{proof}
    Since we need to place $n$ rectangles along the border of $X$, the total vertical length required is $nb=nZ^2$. 
    Since $a_i\geq 1$ for every $i\in [n]$, we can upper bound $n$ by $Z$, leading to an upper bound of $Z^3$ for the total length. 
    Since the side length of $X$ is $Z^5$ and $Z\geq 27$, it follows that all rectangles can be separated.
    Choosing $h=\frac{Z^5-nb}{n+1}$ ensures a valid spacing between the rectangles, allowing them to be equally distributed along the border of $X$.  
\end{proof}

\ConstructionThird*

\begin{proof}
Since our almost-square $X$ has side length $N=Z^5$, we need to show that $\bar{l_i}< Z^5$ always holds.
We have chosen $b=Z^2$, and since $b\cdot (l_i+\bar{l_i}) =a_i\leq Z$, it follows that $\bar{l_i}=\frac{a_i}{b}-l_i \leq \frac{Z}{Z^2}-l_i \leq \frac{1}{Z}$.
Since $Z\geq 27$, we have $\frac{1}{Z}\leq \frac{1}{27}<Z^5$, so $\bar{l_i}<Z^5$ always holds.
\end{proof}

\SingleRectangle*

\begin{proof}
Since from \cref{lemma:SingleCycle_x} we know that for $\alpha>1$, the $\alpha$-circularity of multiple disconnected polygons is always worse than that of the single best among them, and since all rectangles are pairwise non-adjacent, it suffices to bound the $\alpha$-circularity score of each individual rectangle.
We can upper bound the $\alpha$-circularity of a rectangle $R_i$ as follows:
\begin{align*}
    c_\alpha(R_i) = \frac{a_i}{\left(2l_i+2\bar{l_i}+2b\right)^\alpha}\leq \frac{Z}{b^{\alpha}} = \frac{1}{Z^{2\alpha - 1}}.
\end{align*}
For a solution $S_I$ with $\sum_{i\in I}a_i=Z/2$, we get an $\alpha$-circularity score of 
\begin{align*}
    c_\alpha(S_I)=\frac{\Ax+Z/2}{\left(\Px+\frac{Z/2}{c}\right)^\alpha}
    =\frac{Z^{10}-Z/2}{\left((4Z^5 + 1)+ \frac{(Z/2)\cdot (4Z^5+1)}{\alpha (Z^{10}-Z) + Z/2(\alpha-1)}\right)^\alpha}.
\end{align*}
By \cref{lem:function_maximum_hardness} we know that any other solution $S'\subseteq \{R_1,\ldots,R_n,X\}$ with $X\in S'$ will have a strictly less $\alpha$-circularity score.
Therefore, for $S_I$ to be the global optimal solution we need the following inequality to hold:
\begin{alignat*}{2}
    c_\alpha(R_i)\leq \frac{1}{Z^{2\alpha - 1}} < \frac{Z^{10}-Z/2}{\left((4Z^5 + 1)+ \frac{(Z/2)\cdot (4Z^5+1)}{\alpha (Z^{10}-Z) + Z/2(\alpha-1)}\right)^\alpha}= c_\alpha(S_I)\\
    \Leftrightarrow  \left((4Z^5 + 1)+ \frac{(Z/2)\cdot (4Z^5+1)}{\alpha (Z^{10}-Z) + Z/2(\alpha-1)}\right)^\alpha < Z^{2\alpha-1}(Z^{10}-Z/2).
\end{alignat*}
Since $Z\geq 27$ and $\alpha\in (1,2]$, we can upper bound the left-hand side by 
\begin{align*}
    \left((5Z^5) + \frac{5Z^6}{Z^{9}}\right)^2 = 25Z^{10} + 50Z^2 + \frac{25}{Z^6}.
\end{align*}
Additionally, we can lower bound the right hand side by 
\begin{align*}
    Z^{11}-Z^2/2.
\end{align*}
Together, we obtain the inequality
\begin{alignat*}{2}
   & 25Z^{10} + 50Z^2 + \frac{25}{Z^6} &&< Z^{11}-Z^2/2 \\
   \Leftrightarrow~& 25 + \frac{50}{Z^8} + \frac{25}{Z^{16}} &&< Z - \frac{1}{2Z^8}.
\end{alignat*}
We can observe that this inequality always holds for $Z\geq 27$. 
Therefore, the $\alpha$-circularity score of any individual rectangle is always worse than the solution in which the corresponding items sum up to $Z/2$.
\end{proof}

\NPhardness*

\begin{proof}
We apply our reduction from the partition problem. 
Let $Y=(a_1,\ldots, a_n)$ be an arbitrary partition instance, and let $Z=\sum_{i\in [n]}a_i$. 
We assume $Z\geq 27$, otherwise the instance can be solved in constant time. 
For any $\alpha \in (1,2]$, we construct our instance as described in the previous procedure. 
As shown in Lemmata~\ref{lem:construction_1} to \ref{lem:construction_3}, this construction always produces a feasible instance.
By Lemmata~\ref{lem:function_maximum_hardness} and \ref{lemma:NP-hardness_single_rectangle}, we know that the largest achievable $\alpha$-circularity score is $\frac{A_X+Z/2}{\left(P_X+\frac{Z/2}{c}\right)^\alpha}$, and this score can only be achieved by selecting a set $S_I=R_I\cup \{X\}$, where $R_I\subseteq \{R_1,\ldots, R_n\}$. 
Furthermore, by Lemma~\ref{lem:function_maximum_hardness} we know that the areas of the rectangles in $S_I$ must sum up to $Z/2$ to achieve the highest $\alpha$-circularity score.
Since the area of an arbitrary rectangle $R_i$ corresponds to $a_i$, for reconstructing the equivalent values for the partition instance, it is enough to select all values $a_i$ such that $R_i\in R_I$.
If the optimal $\alpha$-circularity score is less than $\frac{A_X+Z/2}{\left(P_X+\frac{Z/2}{c}\right)^\alpha}$, then no solution $I\subset [n]$ with $\sum_{i\in I}a_i=Z/2$ can exist. 
Otherwise, we would be able to achieve a strictly better $\alpha$-circularity score.
\end{proof}

\section{Experiments}\label{sec:exp}

As discussed in Section~\ref{sec:Alg}, the algorithm of Park and Phillips~\cite{PP93} can be adapted directly to compute a cycle maximizing $c_\alpha$. For this, one first needs to build the expanded search graph (as discussed after Lemma~\ref{lem:weight_enclosed}). Assume that the optimal cycle contains some vertex $u$ and encloses area $B$.
By the same reasoning as for our algorithm, choosing starting vertex $(u,0)$ and computing the shortest path to vertex $(u,B)$ in the expanded search graph will return the boundary of an optimal cycle. Hence, one could run the algorithm for every possible choice of $u$ and $B$. This inherently results in a running time depending on $\Am$ due to the number of different choices for $B$.

In the following we will experimentally test how this algorithm and \cref{alg:find_pareto_cycles} perform when compared to each other on either synthetically generated datasets or subsets of real-world datasets.
For this purpose we implemented the algorithm of Park and Phillips as described in their original work, with three important practical modifications:
For the expanded graph, we employ a lazy state generation strategy, meaning that vertices are only created when they are actually visited during the search. 
Since a large fraction of all theoretically possible states is never explored, this significantly reduces both memory consumption and runtime. 
Additionally, we prevent cycles that visit the source node multiple times by ensuring that Dijkstra's algorithm never inserts a state corresponding to the source node into the priority queue. 
Finally, as previously mentioned and further discussed in \cref{sec:weight_bounds_instance}, we adapt the weight bounds used in the algorithm to the interval $\{-nW,...,nW\}$ instead of the interval $\{-W,\ldots,W\}$. 
In fact, this modification turned out to be necessary in practice, as the original bounds fail to recover the correct solution on our test instance, whereas the adjusted bounds yield the correct results.

To compare the performance of both algorithms, we measure the number of operations performed: 
For \cref{alg:find_pareto_cycles}, the main bottleneck is the repeated computation of the Pareto-sum in Line 12 in \cref{alg:subroutine_bellmann}.
As discussed in the proof of \cref{thm:algo_cycles}, we can compute the Pareto sum of two Pareto sets $\PP_1,\PP_2$ in $\OO(\vert \PP_1 \vert + \vert \PP_2 \vert)$ if both sets are sorted (e.g., increasing in the length of the solution) by using a simple sweeping approach. 
Therefore, we define the number of operations as the sum of $|L_b| + |L_b'|$ over all executions of Line 12 in \cref{alg:subroutine_bellmann}.

\begin{figure}
    \centering
    \fbox{
    \includegraphics[width=.7\linewidth, trim = 4cm 2cm 2cm 2cm, clip]{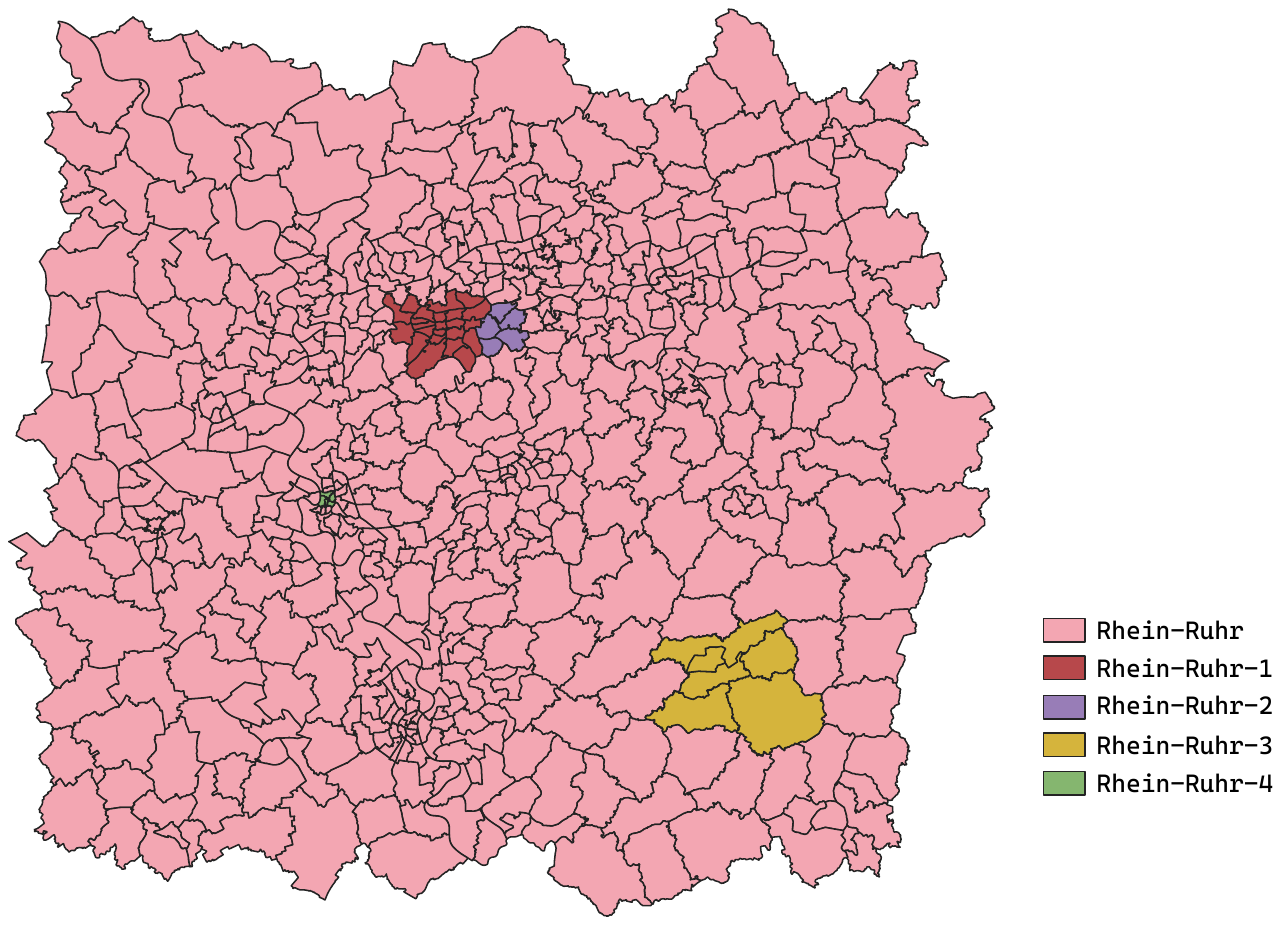}
    }
    \caption{Visualization of the different subsets chosen from the Rhein-Ruhr metropolitan region.}
    \label{fig:dataset_rheinruhr_subsets}
\end{figure}

For the algorithm by Park and Phillips we estimate the computational effort by counting the number of vertices that are actually visited in the expanded graph by Dijkstra's algorithm. 

\subsection{Scaling}
We first conducted experiments on a synthetic test instance.
For simulating the setting of scaling resulting from zooming in or out, we scale the instance by a factor $s\in \NN$, which results in all edge lengths being multiplied by $s$ and all areas by $s^2$. 
The results show that the number of operations for both algorithms remains constant, even for large scaling factors up to $s = 512$.

We can see that this result holds also if we scale the lengths and areas of all polygons independently from each other.
This is because scaling will result in the same equivalent solutions being found: 
Assume in the following every edge length is multiplied with some factor $s_1\in \RR_{>0}$ and every area is multiplied with some factor $s_2 \in \RR_{>0}$.
In both algorithms, if in some iteration we would normally end up with some intermediate solution representing bicriteria cost $(l,w)\in \RR^2$, we will now end up with solution $(s_1\cdot l, s_2\cdot w)$.
As this change will affect every solution evaluated during the algorithm, the overall behavior of both algorithms remains unchanged, as scaling two bicriteria solutions does not change their dominance behavior and the same equivalent solutions will be evaluated.
This means that the number of operations will remain unchanged for both algorithms.

\begin{figure}[t]
    \centering

    \begin{minipage}{0.45\linewidth}
        \centering
        \includegraphics[width=\linewidth, trim=0cm 2.65cm 0cm 2.5cm, clip]{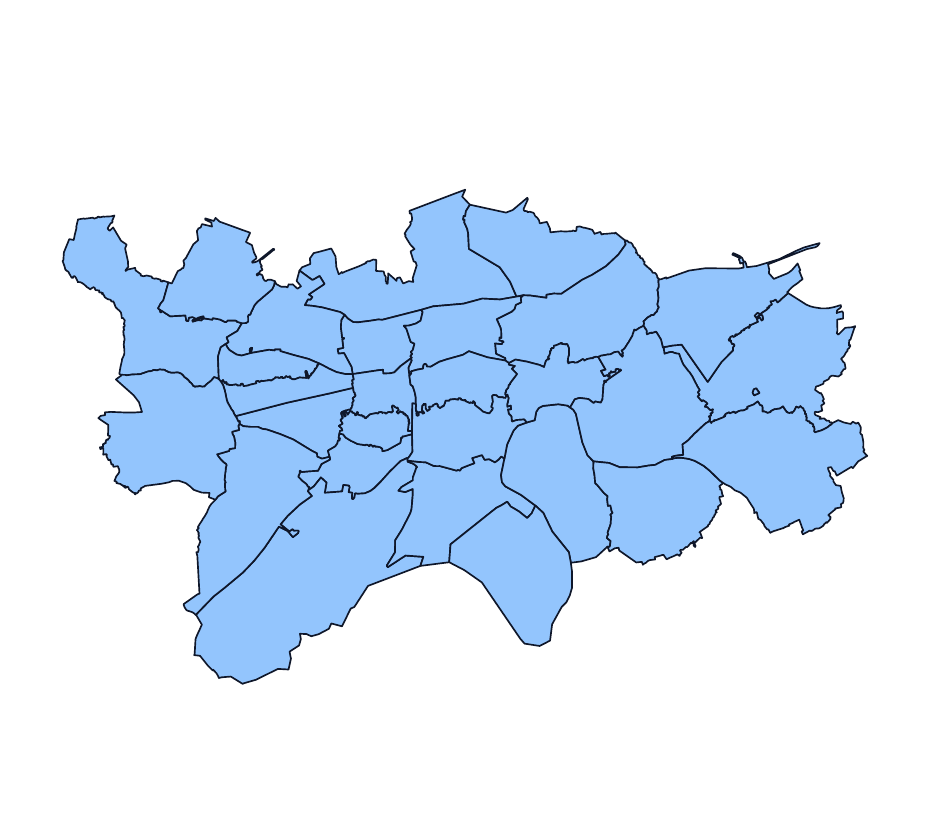}

        \includegraphics[width=0.75\linewidth, trim=1cm 1cm 1cm 1.75cm, clip]{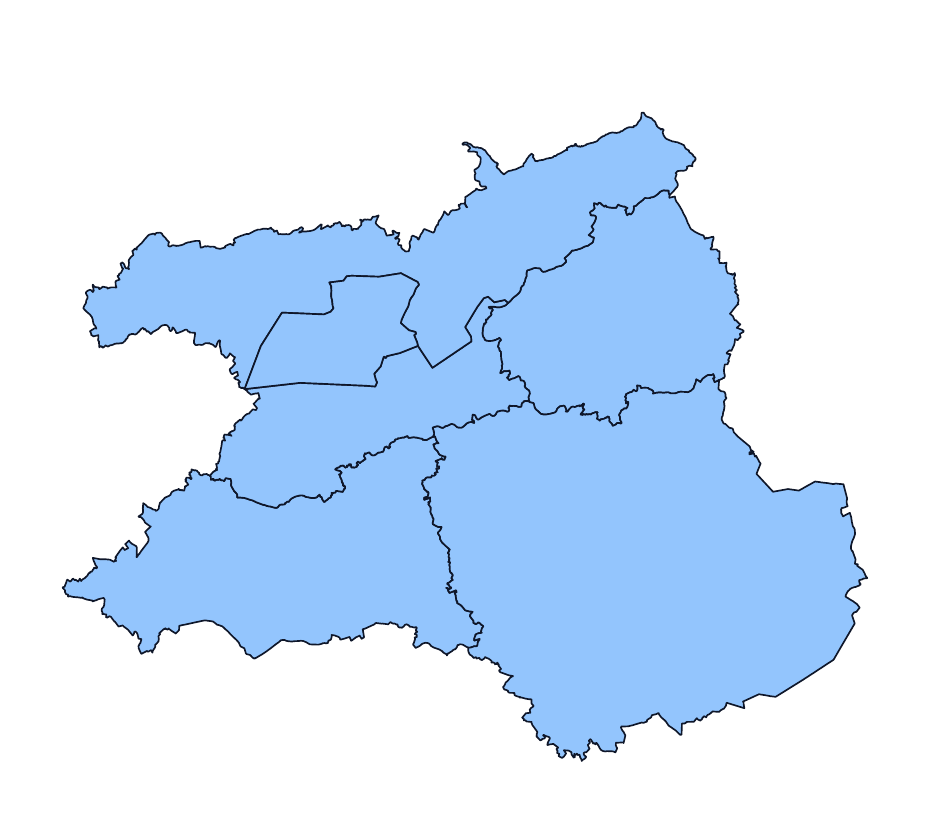}
    \end{minipage}
    \hfill
    \begin{minipage}{0.35\linewidth}
        \centering
        \includegraphics[width=0.8\linewidth, trim=2cm 1cm 2cm 1.5cm, clip]{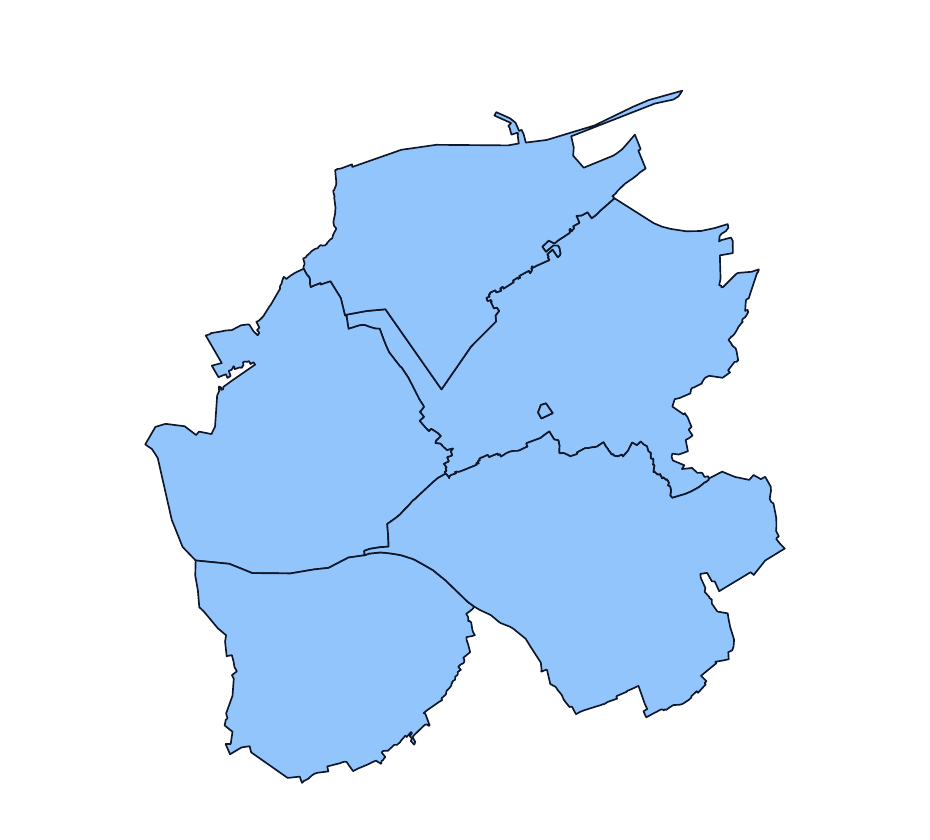}

        \includegraphics[width=0.7\linewidth, trim=1cm 0.5cm 1cm 0.5cm, clip]{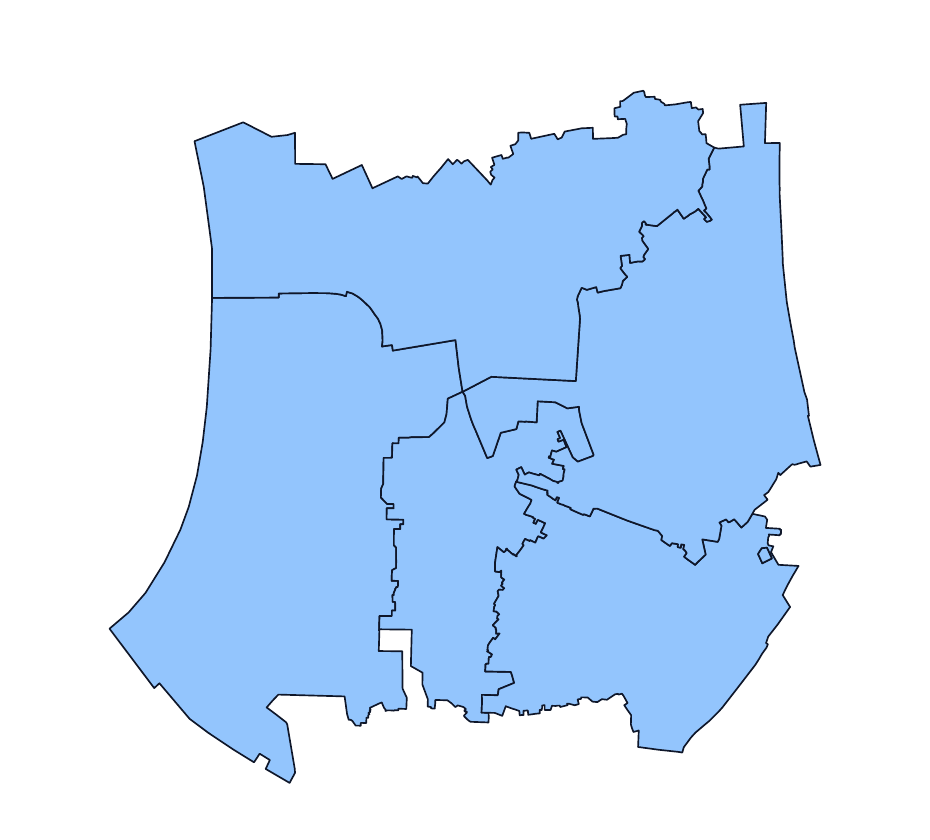}
    \end{minipage}

    \caption{Geometries of the Rhein-Ruhr subsets for \textit{Rhein-Ruhr-1} (top left), \textit{Rhein-Ruhr-2} (top right), \textit{Rhein-Ruhr-3} (bottom left), and \textit{Rhein-Ruhr-4} (bottom right).}
    \label{fig:dataset_shapes}
\end{figure}

\subsection{Real-world datasets}
Because of scaling having no influence on the number of basic operations performed by both algorithms, we will consider the setting of real-world instances that are rounded down for simplification. 

In the following we consider real-world datasets for computing the optimal Polsby-Popper-score.
For this matter we considered 4 connected divisions of the Rhein-Ruhr metropolitan region, labeled \textit{Rhein-Ruhr-i} for $i\in [4]$.
The geometries of the respective Rhein-Ruhr datasets are visualized in \cref{fig:dataset_shapes}. 

In contrast to the previously considered artificially created dataset, these subsets are characterized by significantly larger weights.
For example, while the test instance has a weight bound $W = 106$, the weight bound for \textit{Rhein-Ruhr-2} is approximately $W \approx 4 \cdot 10^8$. 
To make the experiments tractable, we scale the datasets by rounding to a reduced number of significant digits used to represent the weights. 
More explicitly, given some arbitrary value $x\in \RR_{\neq 0}$ and value $i\in \NN$, let
$k_i(x)=10^{\floor{\log_{10}(\vert x \vert)}-i+1}$ be the factor by which $x$ will be scaled such that only the first $i$ significant digits are used.
We define a new weight function 
$w_i(x)=\round{\frac{x}{k_i(x)}}\cdot k_i(x)$, which returns for an arbitrary value $x\neq 0$ a new value such that only the first $i$ significant digits are used and the order of magnitude of both values remains the same.
Assume an edge $e\in E(G_\poly)$ has weight $w(e) = 134082280$.
We now consider approximations such as $w_1(e) = 100000000$, $w_2(e)=130000000$, $w_3(e)=134000000$, and so on. 
We note that the polygons in the Rhein-Ruhr datasets use between 7 and 9 significant digits, so approximations with only 1–4 significant digits remain far from the original precision.

\begin{figure}[t]
    \centering
    \includegraphics[width=0.8\linewidth]{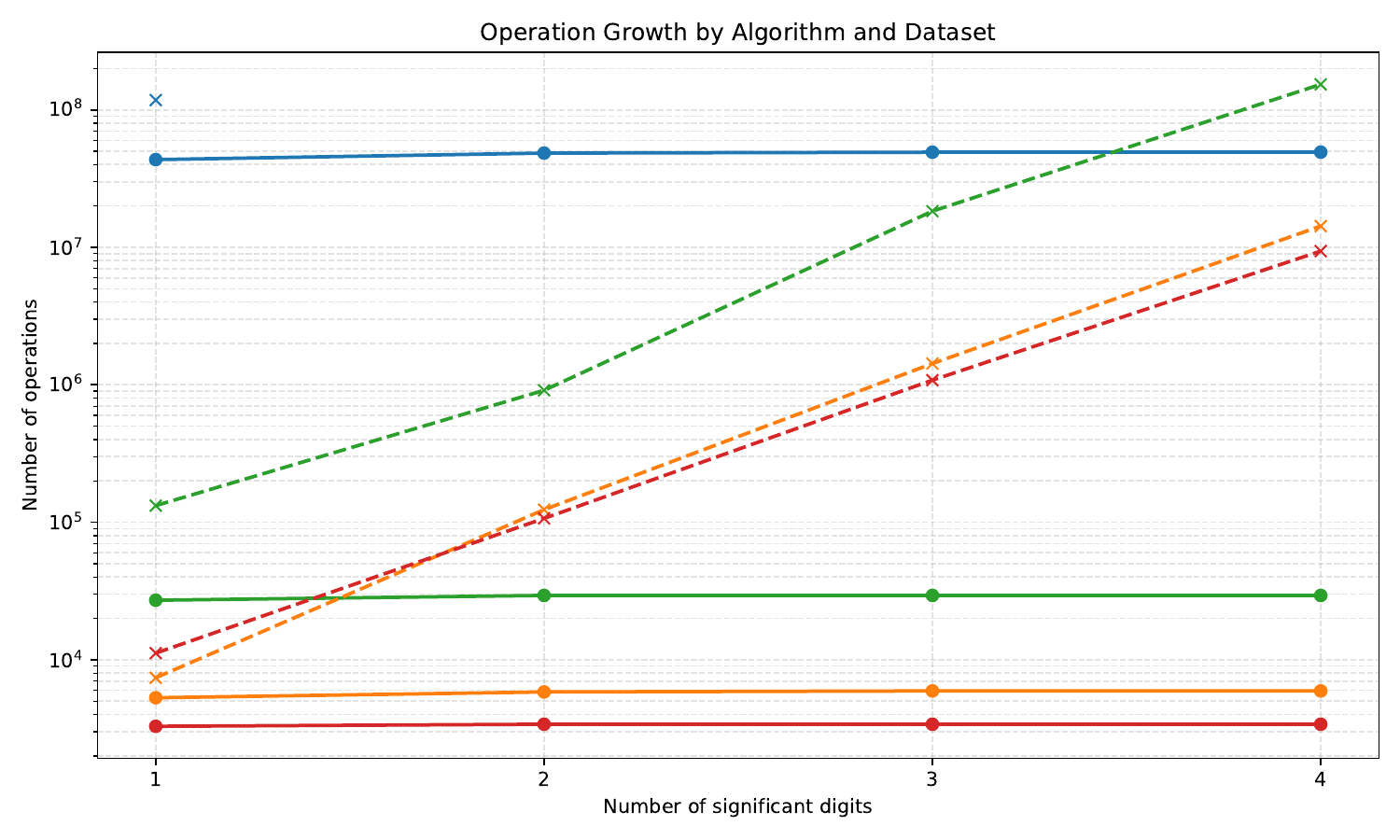}
    \caption{
    Number of operations for both algorithms under varying scaling and precision settings (logarithmic scale on the $y$-axis). Dashed lines correspond to the algorithm by Park and Phillips and solid lines to \cref{alg:subroutine_bellmann}. We indicate the datasets with colors: blue for \textit{Rhein-Ruhr-1}, orange for \textit{Rhein-Ruhr-2}, green for \textit{Rhein-Ruhr-3}, and red for \textit{Rhein-Ruhr-4}. 
For \textit{Rhein-Ruhr-1}, we report results for the Park and Phillips algorithm only for the case of one significant digit, as computations for higher precision levels were prohibitively time-consuming. 
    }
    \label{fig:operation_growth}
\end{figure}
\begin{figure}[t]
    \centering
    \includegraphics[width=0.8\linewidth]{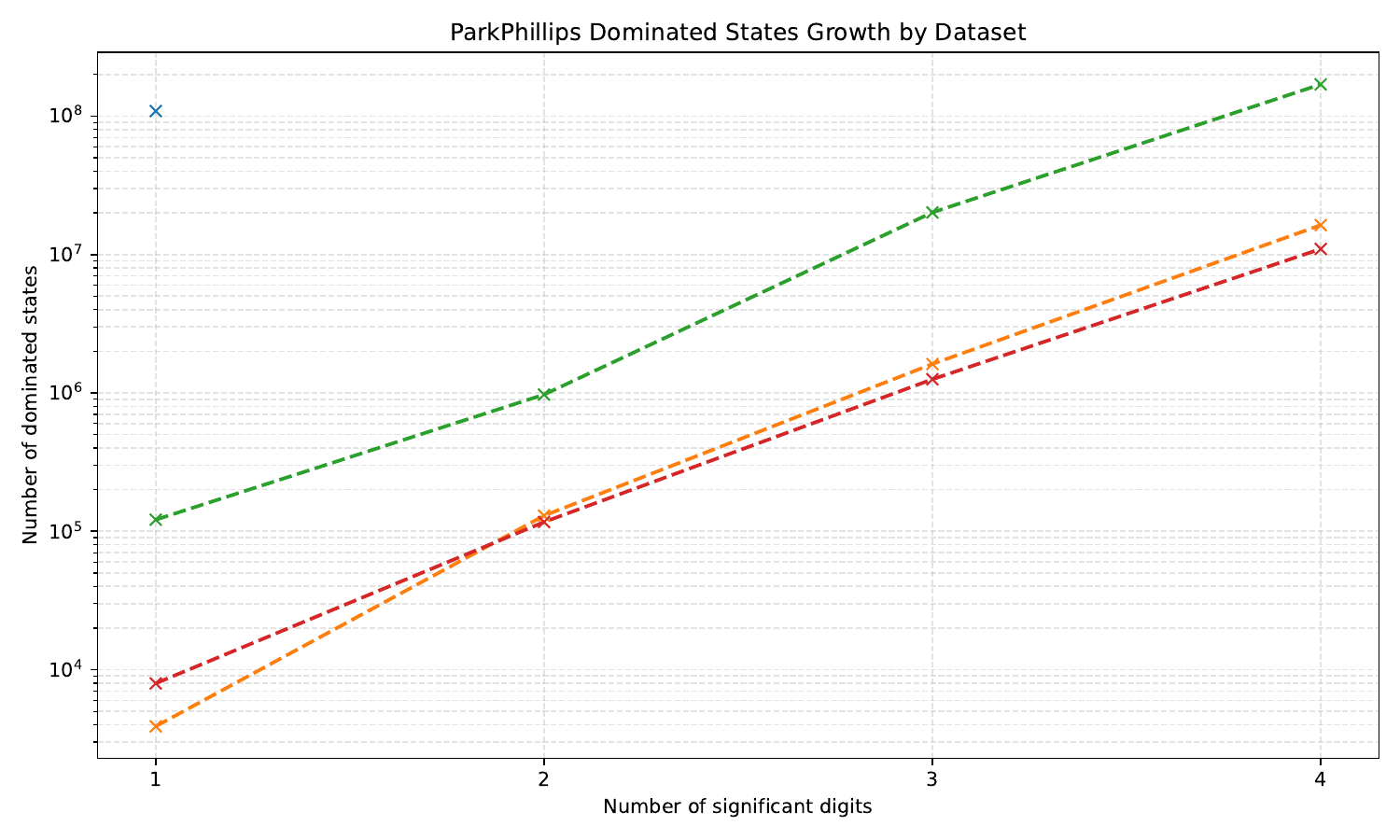}
    \caption{
    Number of dominated states visited by the Park and Phillips algorithm (logarithmic scale on the $y$-axis).
    The datasets are indicated by the same colors as in \cref{fig:operation_growth}.
    }
    \label{fig:dominated_states}
\end{figure}

\subsection{Results}
Across the Rhein-Ruhr subsets, we observe a very small number of Pareto-optimal solutions. Specifically, \textit{Rhein-Ruhr-2}, \textit{Rhein-Ruhr-3}, and \textit{Rhein-Ruhr-4} yield 4, 3, and 5 solutions, respectively.
This number remained stable across all tested levels of precision. 
The only exception occurs for \textit{Rhein-Ruhr-2} with a single significant digit, where one solution is lost due to rounding, resulting in only 3 Pareto-optimal solutions. 
Overall, this indicates that, apart from isolated rounding effects, the number of Pareto-optimal solutions in this dataset appears to be largely independent of the number of significant digits.

Finally, we visualize the number of operations for both algorithms in \cref{fig:operation_growth}. 
A key observation from \cref{fig:operation_growth} is the fundamentally different scaling behavior of the two approaches. 
\cref{alg:subroutine_bellmann} exhibits an almost constant number of operations across all tested levels of precision, demonstrating a high degree of robustness with respect to increasing weight resolution. 
In contrast, the algorithm by Park and Phillips indicates an exponential growth in the number of operations as the number of significant digits increases, which appears as linear growth in the logarithmic plot. 
This highlights a substantial practical advantage of \cref{alg:subroutine_bellmann} on instances with large or high-precision weights.

We further analyze the behavior of the algorithm by Park and Phillips by examining the number of dominated vertices encountered in the expanded graph during the execution of Dijkstra's algorithm. 
Specifically, given some starting vertex $(u,0)$, we count for every reached vertex $(v, w)$ at distance $d_{v,w}$ the cardinality of reached vertices $(v, w')$ at distance $d_{v,w'}$ such that the latter vertex is dominated by the first vertex, i.e., 
$(d_{v,w},w)\prec (d_{v,w'},w')$.
Intuitively, these states are suboptimal with respect to the Pareto dominance relation and therefore are not relevant for the final set of solutions, yet they are still explored during the search.

The results are shown in \cref{fig:dominated_states}. 
We observe a clear exponential increase in the number of dominated states as the precision of the weights increases. 
This mirrors the growth pattern observed for the total number of visited states in \cref{fig:operation_growth} and provides a direct explanation for the poor scaling behavior of the algorithm by Park and Phillips. 
A large fraction of the computational effort is spent exploring dominated states that are ultimately discarded. 
In contrast, \cref{alg:subroutine_bellmann} avoids this combinatorial explosion by construction, as it operates directly on Pareto sets and does not generate dominated intermediate states. 
\section{Weight bounds for intermediate paths}\label{sec:weight_bounds_instance}

In this section we will show that there exist instances such that the absolute weight of any optimal simple tour in the search graph cannot be bounded by any constant factor of $W$.

For the minimum quotient cut problem, we are given a graph $G=(V,E)$ with vertex weights $w:V\rightarrow \mathbb{R}_{\geq 0}$ and edge weights $c:E \rightarrow \mathbb{R}_{\geq 0}$.
We aim to find a cut $(S,\bar{S})$ that minimizes the sparsity ratio $\frac{c(S,\Bar{S})}{\min\{w(S),w(\Bar{S})\}}$. 
Let $W=\sum_{v\in V}w(v)$ denote the total weight of all vertices.

In the following we consider the instance and its dual illustrated in \cref{fig:w:construction}, which consists of a constant number of blue vertices ($c_b$), orange vertices ($c_o$), green edges ($c_g$), and dotted edges ($c_d$).
The two central vertices have weight $B/2$, where $B$ is a large constant. 
All remaining smaller vertices have unit weight.
The edges in \cref{fig:w:original_instance} shown in green have cost 1, while all dotted edges have infinite cost.
Thus, the total weight is $W=B+(c_b-1)+(c_o-1)$.

By selecting all blue vertices as $S$, only the green edges contribute to the cost of the cut.
As there are $c_g$ green edges, we get $c(S,\bar{S})=c_g$.
Additionally, we have $\min\{w(S),w(\bar{S})\}= B/2 + \min\{c_b-1,c_o-1\}$.
This gives us a total score of $\frac{c(S,\bar{S})}{\min\{w(S),w(\bar{S})\}}=\frac{c_g}{B/2 + \min\{c_b-1,c_o-1\}}$.
As we have a constant number of vertices and edges, depending on the choice of $B$ this value becomes arbitrarily small.
We can also see that in any case this is the unique optimal solution, as any proper subset of the set of edges with cost $1$ will not result in two sets being cut from each other, so any different valid cut must  cut at least one edge with cost $\infty$.

\begin{figure}
    \begin{minipage}[t]{0.475\textwidth} 
        \centering
        \includegraphics[width=\textwidth, page=7]{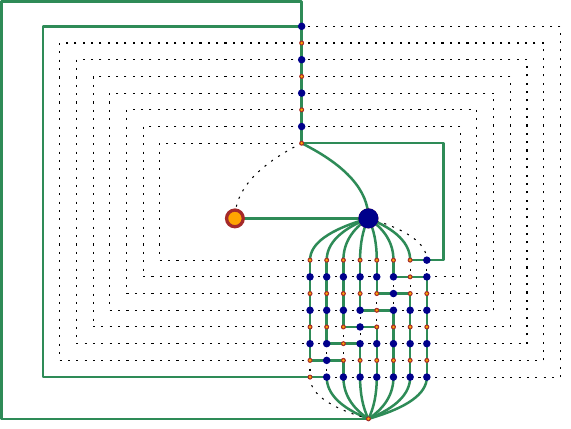} 
        \subcaption{Original planar graph $G$.}
        \label{fig:w:original_instance}
    \end{minipage}
    \centering
    \begin{minipage}[t]{0.475\textwidth} 
        \centering
        \includegraphics[width=\textwidth, page=4]{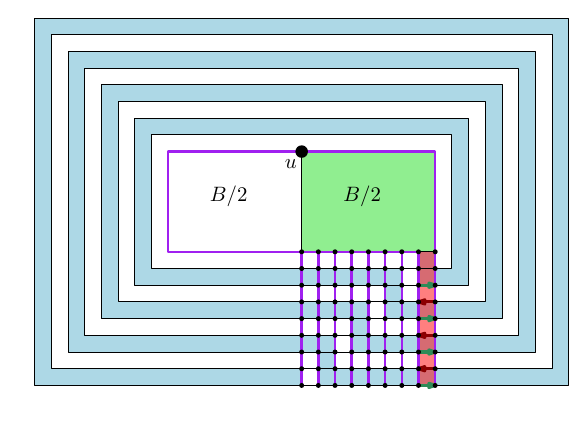} 
        \subcaption{Dual graph.}
        \label{fig:w::dual_of_original}
    \end{minipage}
    \caption{
    Instance construction for weight bound.\\
    \cref{fig:w:original_instance}: Original graph $G$. All dotted lines have an arbitrary high cost and all green edges have a cost of $1$. Removing all green edges separates $G$ into the blue and yellow vertices.\\
    \cref{fig:w::dual_of_original}: Dual graph. Both areas in the middle have same weight $B/2$. The spanning tree for the construction of $G_s$ consists of the violet edges. The red shaded area consists of all edges with absolute weight at least $B$, where green edges are positive and brown edges are negative.
    }
    \label{fig:w:construction}
\end{figure}

Now let us consider its dual graph, shown in \cref{fig:w::dual_of_original}. 
We can see that the unique optimal solution set $S$ corresponds to the areas shaded in green and blue.
Assume the corresponding search graph $G_s$ is constructed from a spanning tree shown as violet edges. 
Then, by construction, all edges shown in green and brown will have an absolute value at least $B$, where green edges are positive and brown edges are negative.
In the following we will denote all edges whose associated weight depends on the choice of $B$ as \emph{heavy edges} and all remaining edges as \emph{light edges}.
By our choice of the spanning tree, the set of heavy edges will correspond to the marked edges in the red area and the border between the two polygons with weight $B/2$ each.

As the optimal solution is unique, it will always be a cycle corresponding to the boundary of the colored polygon.
Starting from vertex $u$, when following along the border of the respective shaded polygon the traversal encounters the heavy edges $-B/2,B+2,B+4,B+6,B+8,-B-7,-B-5,-B-3$, which, together with the corresponding light edges, will sum up to $B/2+36$.
Now, if we only consider the heavy edges, any tour that describes the boundary of the shaded polygon in counterclockwise order and starts at an arbitrary vertex of the border will have at some point an intermediate sum of at least $2B$.
As all other edges can only contribute at most a constant value to the accumulated cost, it follows that if $B$ is chosen large enough, regardless of the starting point of the tour, the accumulated absolute weight will be larger than $W=B+(c_b-1)+(c_o-1)$.

In our example instance, the accumulated weight bound was only a constant factor larger than the given weight bounds $-W$ and $W$.
As one can simply generalize this instance to incorporate more spirals in the dual graph, we can see that the accumulated weight for the optimal cycle cannot be bounded by any constant factor of $W$.
Therefore, to guarantee the retrieval of the optimal sparsest cut in the search graph, one needs to compute a shortest path in the expanded graph where accumulated weights are bounded by $-n\cdot W$ and $n\cdot W$ instead.

\section{\NP{}-hardness of the \kk{}-Districting Problem}\label{sec:NP-hardness-districting}

Now we consider the $k$-\textit{districting} problem and show that it is weakly \NP{}-hard as well.
In this problem, we are given a set of polygons $\poly=\{p_1,\dots,p_n\}$ and a value $k\leq n$.
Our goal is to partition these polygons into $k$ non-empty groups such that the smallest $\alpha$-circularity score over the $k$ districts is maximized.
More formally, we seek to find an assignment $a:\poly\rightarrow [k]$ such that 
$\min_{i\in [k]}{c_\alpha(S_i)}$ is maximized, where $S_i=\bigcup_{p_j\in \poly:a(p_j)=i}p_j$ is the union of polygons assigned to set $S_i$, and $S_i\neq \emptyset$ for all $i\in [k]$.

Using our previous construction in a slightly modified way, we can show that for any $k\geq 2$ and $\alpha\in (1,2]$, this problem is weakly \NP{}-hard as well. 
We use the same reduction from the partition problem and construct our instance in nearly the same manner  as for the original \NP{}-hardness proof. 
The only difference is that, for $k=2$, we also include a very large almost-circle $C$ which is drawn as if it is slightly beneath the almost-square $X$ as well as all rectangles $R_1,\ldots,R_n$. 
For a visualization, see Figure~\ref{fig:hardness_districting}.
For $k\geq 3$, we also add $k-2$ shapes $Y_1,...,Y_{k-2}$, which are adjacent to either only the large almost-circle $C$ or to one another at 
a single point.
For a visualization see Figure~\ref{fig:hardness_districting_k}.
Each of the shapes $Y_1,\dots,Y_{k-2}$ has an $\alpha$-circularity score of $c_\alpha(S_I)$, where $S_I$ is a partition shape solution as in \cref{sec:NP-hardness}.

\begin{figure}
    \centering
    \begin{minipage}[t]{0.475\textwidth} 
        \centering
        \includegraphics[page=2,width=.8\linewidth]{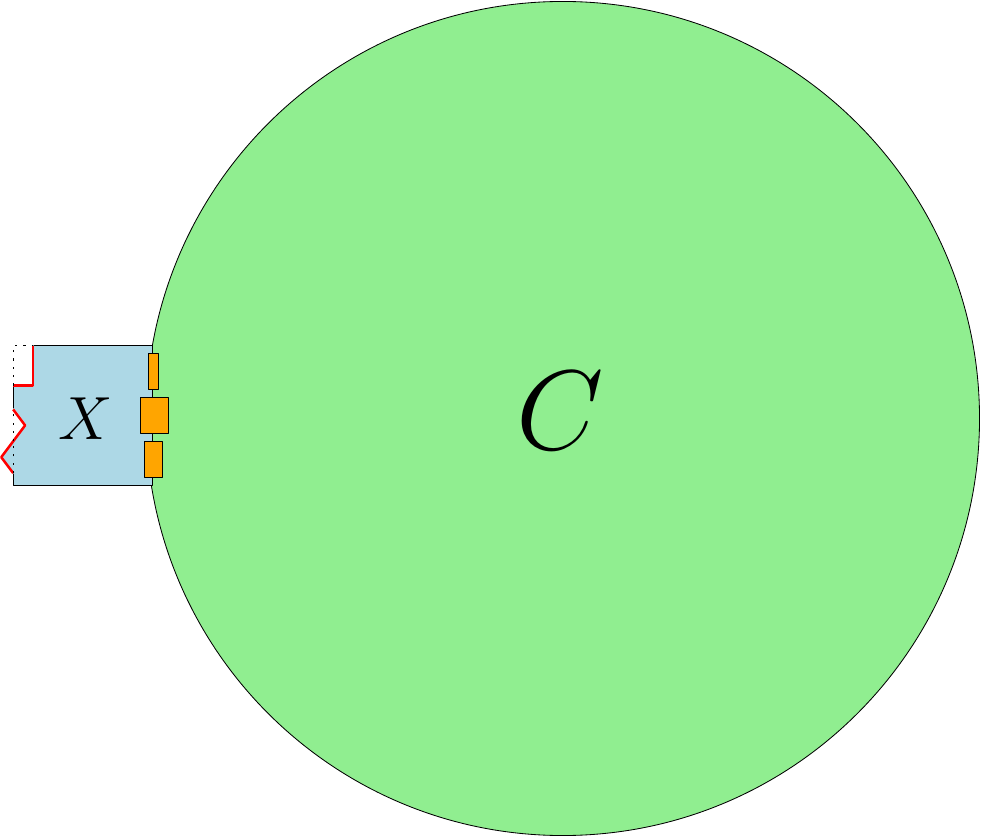}
\caption{Adapted instance for the $2$-districting problem. We include a very large circle $C$ next to our original instance from Section~\ref{sec:NP-hardness}.}
\label{fig:hardness_districting}
    \end{minipage}
    \hfill
    \begin{minipage}[t]{0.475\textwidth}
        \centering
        \includegraphics[page=3,width=.7\linewidth]{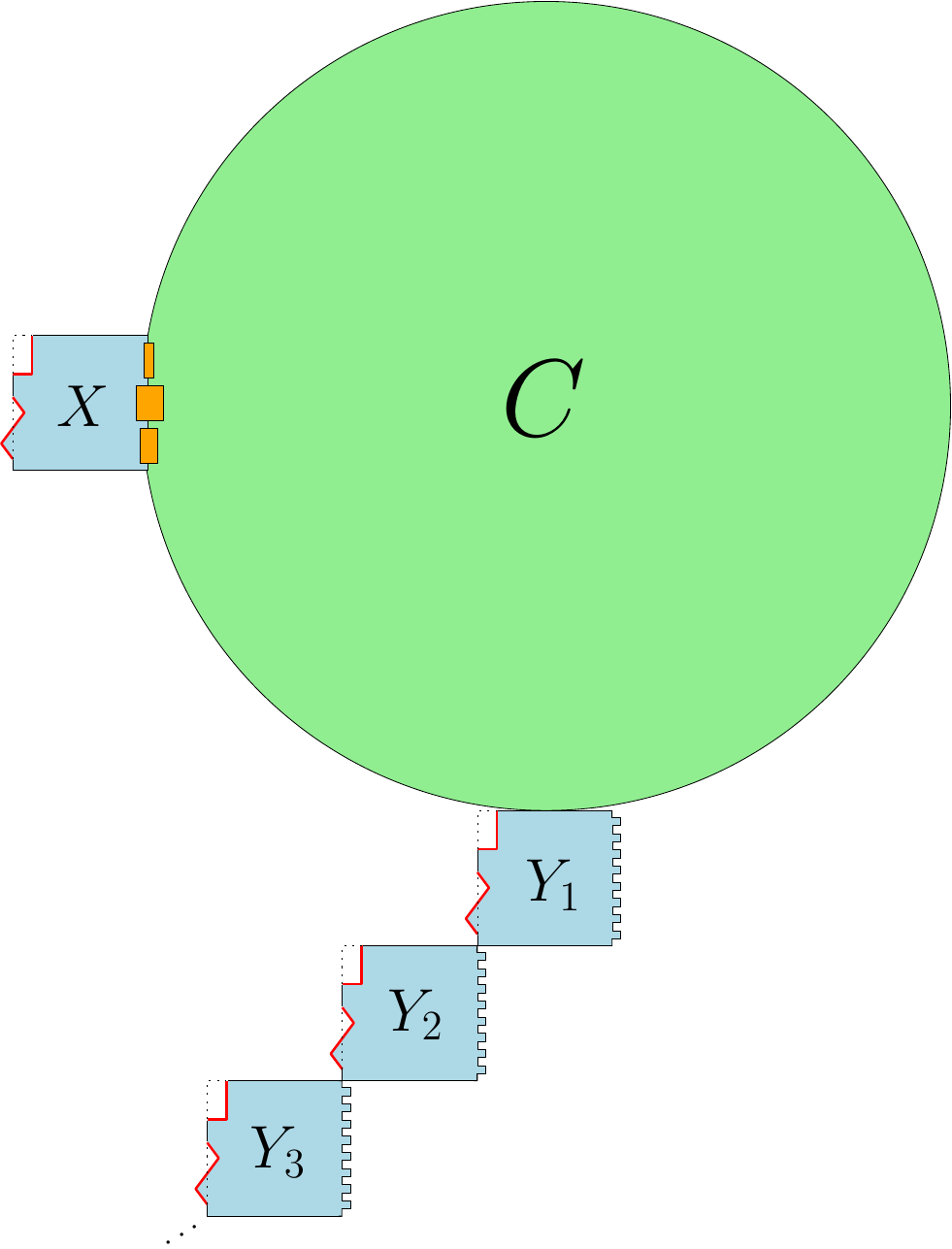}
\caption{Adaptation for the $k$-districting problem. The almost-squares $Y_1$ to $Y_{k-2}$ are designed such that each has area $\Ax+Z/2$ and perimeter $\Px+\frac{Z/2}{c}$.}
\label{fig:hardness_districting_k}
    \end{minipage}
\end{figure}

Intuitively, for any $k\geq 2$ the instance is designed such that the optimal $\alpha$-circularity score can only be achieved as follows:
\begin{itemize}
    \item If $I\subset [n]$ corresponds to a partition solution, assign all rectangles in $R_I$ as well as the almost-square $X$ to the first district.
    \item Assign the remaining rectangles $\overline{R_I} = \{R_1,\dots,R_n\}\setminus R_I$ along with the almost-circle $C$ to the second district.
    \item Each shape $Y_i$ for $i \in [k-2]$ is assigned to one of the remaining $k-2$ districts each.
    \item The $\alpha$-circularity score of any optimal solution is upper bounded by $\frac{A_X+Z/2}{(P_X+ Z/(2c))^\alpha}$ and only holds with equality if a feasible partition solution exists.
\end{itemize}

All districts, except the one containing $C$ (which has an even larger $\alpha$-circularity score), achieve an $\alpha$-circularity score of $c_\alpha(S_I)$.
If any rectangle is assigned to a district that does not contain either the almost-circle $C$ or the almost-square $X$, then, by \cref{lemma:NP-hardness_single_rectangle} and \cref{lemma:SingleCycle_x}, it will always have a worse score than $c_\alpha(S_I)$, because the rectangles can only be assigned to the shapes $Y_1,\dots,Y_{k-2}$ or must form their own district.
Thus, assuming that all rectangles need to be assigned to $C$ and $X$, the only remaining shapes to be assigned to districts are $C,X$ and $Y_1,\dots,Y_{k-2}$, a total of $k$ shapes. 
Since by definition each district needs to contain at least one polygon, it follows that each of the previous shapes must be assigned to a different district.

\begin{restatable}{theorem}{DistrictingHardness}
\label{thm:NP-hardness_districting}
The $k$-\textit{districting} problem is weakly \NP{}-hard for every $k\geq 2$ and $\alpha\in (1,2]$.
\end{restatable}

\begin{proof}
We first argue why this holds true for $k=2$ and then adapt the argument for an arbitrary $k\geq 3$.

For $k=2$, consider the instance shown in \cref{fig:hardness_districting}. 
For a given partition instance we first construct our instance as in \cref{sec:NP-hardness}. 
In addition, we include a large circle $C$, that is drawn such that it appears to be slightly beneath the original instance and intersects the two right corners of the almost-square $X$. 
Because the circle is not a full circle, we will refer to it in the following as the \textit{almost-circle} $C$.
Let $I\subset [n]$ be an optimal index set for the respective partition instance.
This set is also optimal for optimizing $\alpha$-circularity if the corresponding set of rectangles is included with the almost-square $X$. 
We claim that this set of polygons is also optimal for the $2$-districting problem, if the almost-circle $C$ is sufficiently large. 
Let $R_I$ denote the optimal set of rectangles for our original instance, and let $\overline{R_I}$ represent the set of remaining rectangles, i.e., $\overline{R_I} = \{R_1,\dots,R_n\}\setminus R_I$. 
Assume we assign $a(p)=1$ for every $p\in R_I$, as well as $a(X)=1$. 
For every remaining rectangle $p\in \overline{R_I}$ and the almost-circle $C$, we set their assignment to $2$.
We observe that, as we increase the radius of $C$, the score $c_\alpha(S_2)$ will also increase, approaching the largest possible value of $\frac{1}{4\pi}$ for $\alpha=2$, and becoming arbitrarily large for $\alpha \in (1,2)$, 
assuming that the remaining polygons do not change.
Therefore, we can conclude that if $C$ is large enough, $\min_{i\in [2]}c_\alpha(S_i)$ must correspond to the cost of an optimal solution for the $\alpha$-circularity problem variant in Section~\ref{sec:NP-hardness}, i.e., $\min_i c_{\alpha}(S_i) = \frac{\Ax+Z/2}{\left(\Px+\frac{Z/2}{c}\right)^\alpha}$.

It remains to show that no other assignment solution results in a better $2$-districting cost. 
Assume the almost-square $X$ as well as the almost-circle $C$ are assigned to the same district, i.e., $a(X)=a(C)$.
Without loss of generality, assume $a(X)=1$. 
Since at least one polygon must be assigned to the second district, there exists at least one rectangle $R_i$ s.t. $a(R_i)=2$. 
From Lemma~\ref{lemma:SingleCycle_x} we know that, since all rectangles are disjoint from each other, if more than one rectangle is assigned to the second district, we can improve the $\alpha$-circularity score by assigning only the best of these rectangles to the second district. 
However, from Lemma~\ref{lemma:NP-hardness_single_rectangle} we know that the $\alpha$-circularity score of a single rectangle will always be worse than the $\alpha$-circularity score of the rectangle solution corresponding to the optimal partition solution, which has the same $\alpha$-circularity score as the solution constructed earlier for the $2$-districting version. 
This implies that in the optimal assignment the almost-circle $C$ and the almost-square $X$ cannot be in the same district.

Finally, for any assignment of rectangles to either the almost-square $X$ or the almost-circle $C$, as discussed above, we must end up with the $\alpha$-circularity score of the almost-square $X$ and the cost of the rectangles assigned to the same district as $X$. 
However, we already know that the solution maximizing the $\alpha$-circularity for the almost-square $X$ corresponds to a subset of the rectangles that matches an optimal partition solution. 
In conclusion, we see that an optimal solution for $2$-districting can be used to solve the partition instance optimally.

Next, we consider an arbitrary value $k\geq 3$. 
We use the same instance as before, but we add $k-2$ shapes $Y_1,\dots,Y_{k-2}$ at the border of $C$ or in such a way that they touch each other at exactly one point (or at some subpart of the edge with an arbitrarily small common boundary).
Alternatively, one can place $Y_1,\dots,Y_{k-2}$ only adjacent to $C$ and increase the size of $C$ until all shapes fit at the border of $C$ without touching each other. 
Each of these shapes is designed to have an $\alpha$-circularity score of $c_\alpha(S_I)$.
We can see that if we assign each polygon $Y_i$ for $i \in [k-2]$ to its own district and use the optimal assignment for $k=2$, the final cost will match the previous $\alpha$-circularity score of an optimal solution for $k=2$, i.e., $\min_i c_{\alpha}(S_i) = \frac{\Ax+Z/2}{\left(\Px+\frac{Z/2}{c}\right)^\alpha}$. 

By Lemma~\ref{lemma:SingleCycle_x}, we know that if any shape $Y_i$ is assigned to a rectangle or the almost-square $X$, but not to the almost-circle $C$, its $\alpha$-circularity score (which is also an upper bound of the overall $\alpha$-circularity score) must strictly decrease. 
Therefore, any such assignment cannot yield a better solution. 
On the other hand, assume some of these shapes are assigned to $C$. 
If we only consider the assignments of all $Y_i$ and $C$ to some districts, there must be at least two districts which did not get assigned any polygon so far. 
Since only the almost-square $X$ and all rectangles $R_1,\dots,R_n$ remain to be assigned, at least one district must be assigned only rectangles. 
By \cref{lemma:NP-hardness_single_rectangle}, this automatically leads to a worse $\alpha$-circularity score for this district compared to $c_\alpha(S_I)$.
Therefore, the previously described assignment with a $k$-districting score of $\frac{\Ax+Z/2}{\left(\Px+\frac{Z/2}{c}\right)^\alpha}$ must be optimal, if a partition solution exists.
\end{proof}

\end{document}